\documentclass[10pt,a4paper]{article}

\usepackage{a4wide}
\usepackage{amsthm}
\newtheorem{theorem}{Theorem}
\newtheorem{lemma}{Lemma}
\newtheorem{proposition}{Proposition}
\newtheorem{corollary}{Corollary}
\usepackage{hyperref} 

\usepackage{xcolor}
\usepackage{graphicx}
\usepackage{amssymb}
\usepackage{amsmath}
\usepackage{stackrel}
\usepackage{cite}
\usepackage{txfonts}
\usepackage{wrapfig}
\usepackage{mathrsfs}
\usepackage{subfigure}

\newcounter{number}

\newenvironment{varitemize}
{
\begin{list}{\labelitemi}
{\setlength{\itemsep}{0pt}
 \setlength{\topsep}{0pt}
 \setlength{\parsep}{0pt}
 \setlength{\partopsep}{0pt}
 \setlength{\leftmargin}{15pt}
 \setlength{\rightmargin}{0pt}
 \setlength{\itemindent}{0pt}
 \setlength{\labelsep}{5pt}
 \setlength{\labelwidth}{10pt}
}}
{
 \end{list} 
}

\newenvironment{varnumlist}
{
\begin{list}{\arabic{number}.}
{\usecounter{number}
 \setlength{\itemsep}{0.0mm}
 \setlength{\topsep}{0.0mm}
 \setlength{\parindent}{0.0mm}
 \setlength{\parskip}{0.0mm}
 \setlength{\parsep}{0.0mm}
 \addtolength{\leftmargin}{-\labelsep}}}
{
 \end{list} 
}

\newcommand{\symone}{\alpha} 

\newcommand{\lblone}{\ell} 
\newcommand{\lbltwo}{\nu}
\newcommand{\lblthree}{\xi}
\newcommand{\netone}{\rho} 
\newcommand{\nettwo}{\pi} 
\newcommand{\netthree}{\theta}
\newcommand{\netfour}{\xi}
\newcommand{\netfive}{\upsilon}

\newcommand{\midd}{\; \; \mbox{\Large{$\mid$}}\;\;} 
\newcommand{\bnf}{::=}

\newcommand{\natone}{n}

\author{Ugo Dal Lago\and Ryo Tanaka\and Akira Yoshimizu}
\title{The Geometry of Concurrent Interaction:\\ Handling Multiple Ports by Way of Multiple Tokens\\(Long Version)}
\date{}

\begin{document}
\maketitle
\begin{abstract}
We introduce a geometry of interaction model for Mazza's
multiport interaction combinators, a graph-theoretic formalism which
is able to faithfully capture concurrent computation as embodied by
process algebras like the $\pi$-calculus. The introduced model is
based on token machines in which not one but \emph{multiple}
tokens are allowed to traverse the underlying net \emph{at the same time}.
We prove soundness and adequacy of the introduced model. The former is
proved as a simulation result between the token machines one obtains
along any reduction sequence. The latter is obtained by a fine analysis of
convergence, both in nets and in token machines.
\end{abstract}

\section{Introduction}
Game semantics~\cite{HylandO00,AbramskyJM00} and the geometry of
interaction (GoI for short)~\cite{Girard89,DanosRegnier} are semantic frameworks in
which programs and proofs are interpreted as mathematical or
computational objects exhibiting nontrivial interactive behaviours
(e.g. strategies~\cite{HylandO00}, token machines~\cite{DanosRegnier},
operators~\cite{Girard89}).  This allows for a number of nice
properties.  First of all, these models can be defined so as to be
\emph{compositional}, but close to \emph{contexts} as for their
discriminating power; this is particularly true in game semantics, for
which many full abstraction results have been proved. Secondly,
interactive models can often be presented concretely, either as
circuits~\cite{GhicaGOS-I,GhicaGOS-II},
automata~\cite{DanosRegnier}, or abstract machines for
strategies~\cite{FredrikssonGhicaLICS2013}, thus enabling direct
compilation of higher-order programs into low-level languages.

The bulk of the huge amount of literature on interactive semantic models is
about sequential languages~\cite{HylandO00,AbramskyJM00} but,
especially in the last fifteen years, the research community has been
able to devise game models sufficiently powerful to
interpret not only advanced features like polymorphisms and general
references~\cite{AbramskyJ03,AbramskyHM98}, but also
concurrency~\cite{AbramskyMellies99LICS,Mellies-AsynchronousGames,RideauWinskelLICS,WinskelLICS2012},
thus going significantly beyond game semantics as originally conceived
(more on this is in Section~\ref{sect:relatedwork} below). The same cannot
be said about the geometry of interaction, which until very recently has
been able to interpret only sequential, although potentially
effectful, forms of computation \cite{HoshinoMH14,MuroyaHH16}.

A crucial observation, which is the starting point of this work, is
that the geometry of interaction, when formulated in terms of so-called token
machines, can be made parallel by allowing \emph{more than} a single token to
float around \emph{at the same time}~\cite{lics2014}. The
consequences of this idea have been analysed for sequential
languages~\cite{lics2015}, also in presence of probabilistic and quantum
effects~\cite{preprintpopl}. Are multiple tokens enough to model 
fully-fledged concurrency, as embodied by process
algebras? This is the question we will try to address in this
paper. The answer will be positive, although the walk to it will not
be easy.

When looking for a GoI semantics for concurrent models of
computation, one could of course proceed by considering any concrete
process algebra, and define a token machine for it \emph{directly},
being inspired by the literature. As an example, the way the
higher-order $\pi$-calculus is classically encoded into the
usual name-passing $\pi$-calculus~\cite{Sangiorgi93} can be seen reminiscent
of the usual GoI construction, since higher-order process passing is
encoded into a somehow more basic, essentially first-order
calculus. This route would however be biased to a specific process
algebra, thus losing in generality and canonicity. In turn, relying on
a (possibly complicated) computational model would mean hiding the
\emph{structure} of the introduced GoI, understanding what is our main
aim here.

For these reasons, we will purposely take a minimalistic approach,
being inspired by Lafont's \emph{interaction nets} (\emph{INs} for
short) and \emph{interaction combinators} (for short,
\emph{ICs})~\cite{Lafont}.  The latter is a system of interaction nets
which is \emph{universal}, and thus embodies a vast class of
graph-theoretic models for sequential interaction, including logical
systems~\cite{GonthierAL92}, programming languages~\cite{Mackie}, and
optimal reduction algorithms~\cite{Lamping90}. ICs are not only very
simple themselves, but also admit an elementary geometry of
interaction model, arguably the simplest of all~\cite{Lafont}. ICs,
however, cannot form the basis on which to build a GoI model of
concurrent computation: they are strongly confluent and simply lack
the mixture of parallelism and nondeterminism which is at the heart of
concurrency, although being a very good model of (low-level)
sequential computation.

A picture similar to the one drawn by Lafont but envisioned with
concurrency in mind is the one due to Alexiev, Mazza, and
coauthors~\cite{phdAlexiev,MazzaThesis,MazzaCONCUR2005,DormanM13}. In
the last twenty years, in particular, concurrent extensions of
interaction nets, called \emph{multiport} interaction nets
(\emph{MINs} for short), have been proved to be powerful enough to
faithfully encode process algebras~\cite{MazzaThesis,MazzaCONCUR2005},
to admit an event-structure model~\cite{MazzaMSCS}, and to be strictly
more expressive than multi\emph{rule} interaction
nets~\cite{DormanM13}. Remarkably, MINs have also been shown to admit
a universality theorem with respect to a specific interaction system,
namely \emph{multiport interaction combinators} (for short,
\emph{MICs})~\cite{MazzaThesis}. In a sense, then, we have the
following equation:
$$
\frac{\mathit{INs}}{\mathit{ICs}}=\frac{\mathit{MINs}}{\mathit{MICs}}.
$$
MINs generalise INs in that cells can have \emph{more than one}
principal port, this way allowing for a very general form of
nondeterministic interaction, and MICs, being universal for MINs, are
a natural candidate model for studying GoI models for concurrent
systems, as suggested by Mazza
himself~\cite{MazzaThesis}. Nevertheless, we are not aware of any
attempt to give geometry of interaction models for any system of MINs.

When trying to define a GoI model for MICs, one immediately realises
that classic token machines are simply \emph{not} adequate to
faithfully model concurrent interaction. In particular, multiport
interaction, i.e., nondeterminism as found in MICs, cannot be captured
by automata in which just one token is allowed to float around the
underlying graph, as will be explained in
Section~\ref{sect:whynotenough} below. The only way out consists in
allowing for the simultaneous presence of multiple
tokens~\cite{lics2014,lics2015}, which becomes essential here. In
particular, it allows for \emph{nondeterministic} and \emph{nonlocal}
interaction, which is an essential ingredient of MICs' dynamics, and
of concurrency in general.

This paper is devoted to presenting the first geometry of interaction
model for multiport interaction combinators. Our GoI model is a
substantial extension of Lafont's classic token machine model for
ICs~\cite{Lafont}. We allow multiple tokens to move around at the same
time and make use of them to realise the process of resolving
nondeterminism and keeping track of choices, that we call
\emph{marriages}. This requires four different kinds of token, three
of them being static, and only one meant to really travel inside the
net. This way, we get a stateless notion of a machine even if, of
course, static tokens could be replaced by stateful cells.

Our \emph{Multi-token machines} can be seen as \emph{locative transition
  systems}, a special class of labelled transition systems. This
framework provides us with a natural way of defining parallel
composition, and we will prove our semantics to be indeed
compositional, i.e., that the parallel composition of two nets can be
interpreted as the parallel composition of the two respective
machines, modulo bisimilarity. We can thus establish soundness of the
model in terms of labelled (bi)similarity. What makes everything much
more complex than in single-token machines is the nondeterministic
nature of the reduction system: a net can make nondeterministic
choices, thus losing the capability of behaving in a certain way. If a
net $\netone$ reduces to another net $\nettwo$ by performing such a
reduction step, then the interpretation of $\netone$ is not
necessarily behaviourally \emph{equivalent} to that of $\nettwo$, but
can rather be ``larger''. As a consequence, soundness of our GoI model
needs to be spelled out in the form of \emph{similarity} which, by the
way, turns out to come from bisimilarity between other associated
states of the two involved machines.

We also show that our model is not too coarse but \emph{adequate}, in
the sense that token machines reflect the convergence behaviour of the
nets they interpret, both in the ``may'' and in the ``must''
sense. Our proof heavily exploits the bisimulation relations we use to
prove our soundness result, which can relate an execution of a token
machine to another one along net reduction.

\subsection{Related Work}\label{sect:relatedwork}
Concurrent (or asynchronous) game semantics, first
introduced in \cite{AbramskyMellies99LICS} and pursued later by Melli\'es
\cite{Mellies-AsynchronousGames}, is a generalisation of usual game
semantics  for sequential computation \cite{HylandO00,AbramskyJM00}.
It yielded a fully abstract model of multiplicative additive
linear logic proofs, followed by (again fully abstract) models of many
concurrent calculi, e.g.\ CSP or Parallel Algol
\cite{Laird01,GhicaM04}.  Recently, Winskel and his coauthors have
studied winning conditions and determinacy results for such games,
that may lead to applications in verification of concurrent systems
\cite{RideauWinskelLICS,WinskelLICS2012}. We are not aware of any
attempt to relate all this to geometry of interaction and token
machines.

Another formalism of concurrent and distributed systems that is worth
mentioning here is the one of Petri nets~\cite{Petri66CommunicationWithAutomata}.
Indeed, our token machines resemble Petri nets to a large extent:
multiple tokens circulate around a graph structure, dynamically
enabling or disabling each other's transition.  However, there is one
remarkable difference: while in Petri nets the underlying graph
consists itself of places and transitions and computation is
inherently \emph{local}, out token machines indeed allow tokens lying
next to cells which are far away from each other in a net to
communicate, meaning that interaction is \emph{non}local.  This
is an inevitable price to pay if we want token machines to properly
reflect the behaviour of the MIC reduction rules.  See
Section~\ref{sect:whynotenough} for the details, and
Section~\ref{sect:discussion} for more observations in this direction.

Differential interaction nets~\cite{EhrhardR05} are a graphical
calculus for differential linear logic~\cite{Ehrhard05}, and can be
seen as a multi\emph{rule} variant of interaction nets, for which a
single token GoI model already exists~\cite{deFalco}.  They exhibit
nondeterministic behaviour and are able to encode finitary fragments
of Milner's $\pi$-calculus~\cite{EhrhardLaurent10}, although the
encoding \emph{cannot} be completely satisfactory, as highlighted by
Dorman and Mazza \cite{DormanM13}.  We chose MICs as our target
calculus with this observation in mind: multirule interaction nets
simply lack the expressive power which is necessary to model
concurrency in its generality.  On the other hand, the solid logical
basis and the accompanying type system of differential interaction
nets may offer us a more structural way to deal with problems in
concurrency theory. The authors believe that studying the nature and
structure of the token-flowing can be a way to devise appropriate type
structures and logical systems for calculi which lacks any of those,
like MICs.  This is a topic the authors are currently working on, but
which lies outside the scope of this paper.

\section{Multiport Interaction Combinators at a Glance}\label{sect:MIC}
This section is devoted to introducing the objects of study of this
paper, namely multiport interaction combinators.

\newcommand{\wrone}{e} 
\newcommand{\wrtwo}{f} 
\newcommand{\wrthree}{g} 
\newcommand{\wrset}{\mathcal{WIR}}
\newcommand{\prone}{p} 
\newcommand{\prtwo}{q} 
\newcommand{\prthree}{r} 
\newcommand{\prset}{\mathcal{POR}}
\newcommand{\fprset}{\mathcal{FPOR}}
\newcommand{\clset}{\mathcal{CEL}}
\newcommand{\ctset}{\mathcal{CT}} 
\newcommand{\redrules}{\mathop{\mathcal{R}}}
\newcommand{\netset}{\mathsf{NETS}}
\newcommand{\ckset}{\mathscr{K}}
\newcommand{\lsone}{\mathcal{L}}
\newcommand{\lstwo}{\mathcal{M}}
\newcommand{\lsthree}{\mathcal{N}}
\newcommand{\setone}{X}
\newcommand{\settwo}{Y}
\newcommand{\elmone}{x}
\newcommand{\elmtwo}{y}
\newcommand{\mst}[1]{\{{\kern-1.5pt}|#1|{\kern-1.5pt}\}}
\newcommand{\msts}[1]{\mathbb{M}(#1)}
\newcommand{\fmsts}[1]{\mathbb{FM}(#1)}
\newcommand{\NN}{\mathbb{N}}
\newcommand{\pinjone}{\sigma}
\newcommand{\inv}[1]{#1^{-1}}
\newcommand{\dom}[1]{\mathit{dom}(#1)}
\newcommand{\rng}[1]{\mathit{rng}(#1)}
\newcommand{\setdiff}{\triangle}
\newcommand{\pc}[3]{#1||_{#2}#3}

\subsection{Nets}
A \emph{cell} is a triple $(\symone, \vec{\prone}, \vec{\prtwo})$ of a
symbol $\symone$, a sequence $\vec{\prone}$ of \emph{auxiliary ports},
the length of which is the \emph{arity} of $\symone$, and a sequence
$\vec{\prtwo}$ of \emph{principal ports}, the length of which is the
\emph{coarity} of $\symone$. A cell is drawn in Figure~\ref{fig:cell}.
\begin{figure*}
  \begin{center}
   \fbox{
     \begin{minipage}{.97\textwidth}
      \hspace{32pt} 
      \subfigure[A Cell]{\label{fig:cell}
      \begin{minipage}[c]{.0843\textwidth}
        \centering
        \includegraphics[scale=0.7]{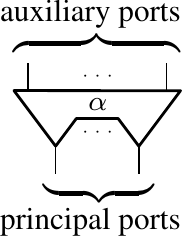}
        \vspace{7pt}
      \end{minipage}}
      \hspace{32pt}
      \subfigure[A Net]{\label{fig:labellednet}
      \begin{minipage}[c]{.099\textwidth}
        \centering
        \includegraphics[scale=0.7]{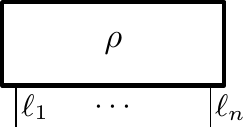}
        \vspace{7pt}
      \end{minipage}}
      \hspace{32pt}
      \subfigure[Interaction Combinators]{\label{fig:mics}
      \begin{minipage}[c]{.199\textwidth}
        \centering
        \includegraphics[scale=0.8]{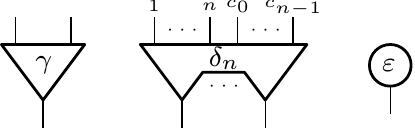}
        \vspace{7pt}
      \end{minipage}}
      \hspace{32pt}
      \subfigure[A Vicious Circle]{\label{fig:vcs}
      \begin{minipage}[c]{.1539\textwidth}
        \centering
        \includegraphics[scale=0.6]{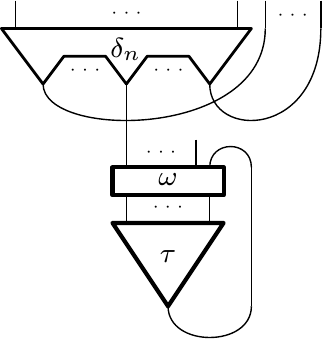}
        \vspace{7pt}
      \end{minipage}}
      \hspace{32pt}
  \caption{Cells, Nets, Combinators, and Vicious Circles.}
  \end{minipage}}
\end{center}
\end{figure*}
A \emph{net} consists of finitely many cells, \emph{free ports} and
\emph{wires} which connect each ports with another one.
Formally, then, given a set of cell kinds $\ckset$, a net $\netone$ is
  $\clset^\netone$ and $\wrset^\netone$ where:
  \begin{varitemize}
  \item
    $\prset^\netone$ is a finite set of ports.
  \item
    $\clset^\netone$ is a finite set of cells of a kind $\symone \in \ckset$, whose ports are elements of $\prset^\netone$.
  \item
    $\wrset^\netone$ is a finite set of unordered pairs of ports.
  \item
  Each port appears at least once in $\wrset^\netone$ and
  at most twice in $\clset^\netone \cup \wrset^\netone$.
  \end{varitemize}
The ports of $\netone$ which appear only once in $\clset^\netone \cup
\wrset^\netone$ are said to be \emph{free}.  The set of free
ports of $\netone$ is $\fprset^\netone$.  Nets are indicated with
metavariables like $\netone$ or $\nettwo$.
Given a set of cell kinds $\ckset$, the set of all nets is
denoted by $\netset_\ckset$.  We assume that each free port $\prone$ of any
net is labelled with a distinct label $\lblone$. Often, nets are
presented graphically, e.g. a net $\netone$ with free ports labelled
with $\lblone_1,\ldots,\lblone_n$ looks as in
Figure~\ref{fig:labellednet}. Metavariables like $\lsone$ and $\lstwo$
stands for finite sets of locations $\{\lblone_1,\ldots,\lblone_n\}$.

In \emph{multiport interaction combinators}, one considers cells of
three different kinds, namely $\gamma$ (with arity $2$ and coarity
$1$), $\delta_n$ (with arity $2n$ and coarity $n$, for each $n \geq
1$), and $\varepsilon$ (with arity $0$ and coarity $1$). These are
pictured in Figure~\ref{fig:mics}.
Thus the set of nets in MICs is $\netset_{\{\gamma, \delta_n, \varepsilon \,|\, n \geq 1\}}$.
A \emph{vicious circle} is a subnet of a net like the one depicted in
Figure~\ref{fig:vcs}, where 
$\tau$ is a tree consisting of $\gamma$ cells, and
$\omega$ is a \emph{wiring}, i.e. consists of wires only.
An example of an MIC net can be found in Figure~\ref{fig:netexample}.
It consists of two $\delta_2$ cells, one $\gamma$ cell, and nine free
ports, connected in the way depicted.
The reason why MICs are a relevant
instance of multiport interaction nets is that they allow for a
Universality Theorem (see \cite{MazzaThesis}, Theorem~6.16, page~209), which
states that \emph{any other} system of multiport interaction nets can
be encoded into MICs. This makes MICs a minimal, but extremely
powerful graph-rewriting formalism, since multiport interaction nets
are well known to be expressive enough to faithfully capture the
dynamics of process algebras, and of the $\pi$-calculus in
particular~\cite{MazzaCONCUR2005,phdAlexiev}. This is in contrast to multirule interaction
nets, like differential interaction nets, which are known to be strictly
less expressive than their multiport siblings~\cite{DormanM13}.

\begin{figure}
  \begin{center}
    \fbox{
      \begin{minipage}{.68\textwidth}
        \hspace{29pt}
        \subfigure[Net]{\label{fig:netexample}
        \begin{minipage}[c]{.20\textwidth}
          \centering
          \includegraphics[scale=0.6]{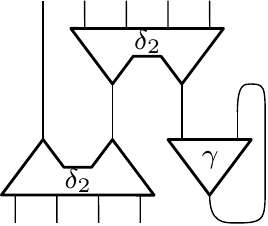}
          \vspace{7pt}
        \end{minipage}}
        \hspace{29pt}
        \subfigure[Reduction]{\label{fig:reductionexample}
        \begin{minipage}[c]{.45\textwidth}
          \centering
          \includegraphics[scale=0.6]{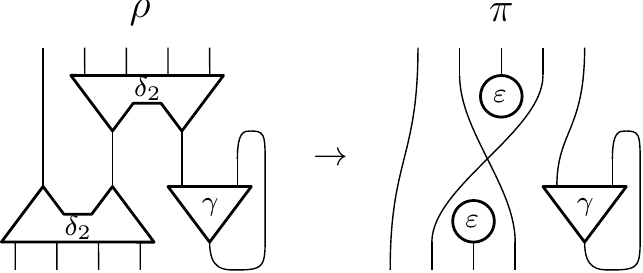}
          \vspace{7pt}
        \end{minipage}}
        \hspace{29pt}
        \caption{An Example}          
    \end{minipage}}
  \end{center}
\end{figure}
\begin{figure}
  \begin{center}
    \fbox{
    \begin{minipage}[c]{.47\textwidth}
      \centering
      \includegraphics[scale=0.8]{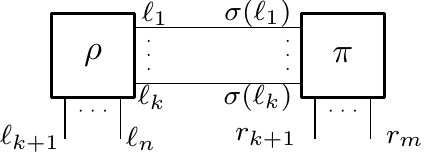}
      \caption{Net Parallel Composition}
      \label{fig:pcnets}
    \end{minipage}}
  \end{center}
\end{figure}

Sometimes it is very convenient to form the \emph{parallel
  composition} of two nets that share certain locations. Any partial
injection $\pinjone$ from a set $\setone$ to a set $\settwo$ can be
seen as a bijection from $\dom{\pinjone}\subseteq\setone$ to
$\rng{\pinjone}\subseteq\settwo$. Its inverse, as usual, is indicated
with $\inv{\pinjone}$. Given nets $\netone$ on $\lsone$ and $\nettwo$
on $\lstwo$ (where $\lsone$ and $\lstwo$ are disjoint) and a partial
injection $\pinjone$ from $\lsone$ and $\lstwo$, the parallel
composition of $\netone$ and $\nettwo$ is defined to be the net
$\pc{\netone}{\pinjone}{\nettwo}$ depicted in Figure~\ref{fig:pcnets},
where $\lsone=\{\ell_1,\ldots,\ell_n\}$ and $\lstwo=\{r_1,\ldots,r_m\}$.

\subsection{Reduction}
\newcommand{\red}{\rightarrow}
\newcommand{\pred}[1]{\rightarrow_{#1}}

The \emph{interaction rules} for the multiport interaction combinators
are the graph rewriting rules (also called \emph{reduction rules}) in
Figure~\ref{fig:redrules}.  We will refer to the three rules as
$\gamma\gamma$, $\delta\delta$ and $\gamma\delta$, respectively.  A
pair of principal ports facing each other in a rule is called a
\emph{redex}.
\begin{figure*}
  \begin{center}
  \fbox{
    \begin{minipage}{.97\textwidth}
      \centering
      \begin{minipage}{.21\textwidth}
      \includegraphics[scale=0.65]{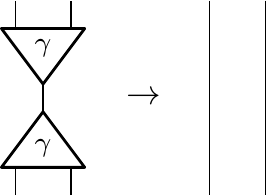}
      \end{minipage}
      \hspace{30pt}
      \begin{minipage}{.37\textwidth}
        \includegraphics[scale=0.65]{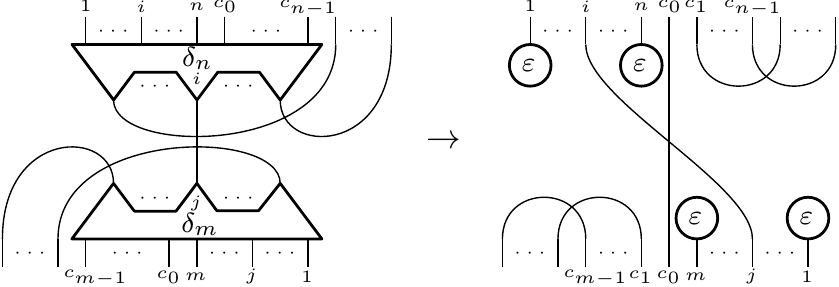}
      \end{minipage}
      \\
      \begin{minipage}{.36\textwidth}
        \includegraphics[scale=0.65]{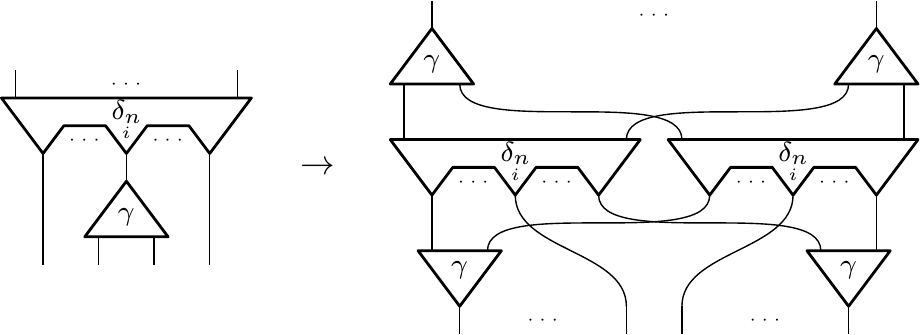}
      \end{minipage}
      \caption{MIC Interaction Rules}\label{fig:redrules}
  \end{minipage}}
  \end{center}
\end{figure*}
Rules can be applied anywhere in a net, giving rise to a binary
relation between nets that we indicate as $\red$.

MICs as defined by Mazza \cite{MazzaThesis} include three more
interaction rules, called $\varepsilon$ rules, which allow to handle
the situation in which an $\varepsilon$ cell faces another cell
through its principal port. All what we say in this paper holds for
the system Mazza considers, and proofs of that can be found in the
Appendix. The choice of considering a simplified system is motivated
by the desire to make our techniques and results easier to
understand, and by the fact that the
Universality Theorem holds also in absence of $\varepsilon$ rules
\cite{MazzaThesis,EV}. For closely related reasons, we decided to
consider a system of multiport interaction combinators in which the
co-arity $n$ of $\delta_n$ cells is always equal to $2$ hereafter in
the paper, and write $\netset$ for
$\netset_{\{\gamma,\delta_2,\varepsilon\}}$.  This, again, allows for
a very simple presentation without sacrificing Universality of MICs.

Consider, as an example, the net in Figure~\ref{fig:netexample}, and
call it $\netone$. It has a redex consisting of two $\delta_2$ cells
facing each other through their principal ports. The net $\netone$ can
thus be rewritten according to the $\delta\delta$ rule, to another net
$\pi$, as described in Figure~\ref{fig:reductionexample}. The net
$\pi$ has no redex, and is thus said to be in \emph{normal form}.
Please observe that a vicious circle occurs in $\netone$, while no
vicious circle can be found in $\nettwo$. This, in other words, is a
witness of the fact that vicious circles are \emph{not} preserved by
reduction, in general. By the way, we consider vicious circles as a
form of a divergent net, given the cyclic chain of $\gamma$ cells they
contain.

\subsection{Why One Token is Not Enough}\label{sect:whynotenough}
\begin{wrapfigure}{R}{.3\textwidth}
\fbox{
\begin{minipage}{.27\textwidth}
\centering
  \includegraphics[scale=1]{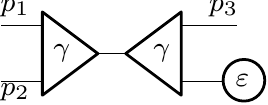}
  \caption{An IC Net.}
  \label{fig:goiexample}
\end{minipage}}
\end{wrapfigure}
In the geometry of interaction, at least in its classic
incarnation~\cite{Girard89,DanosRegnier}, the behaviour of a net is
captured by how the net itself transforms a so-called \emph{token}
when the latter travels inside it. The net, in other words, is seen as
a token transformer, and correctness of the semantics means that
reduction turns nets into \emph{equivalent} ones, i.e., into nets which
behave the same when seen as token transformers.  An example IC net
can be found in Figure~\ref{fig:goiexample}. Seen as a token
transformer, the net sends any token which comes from the free port
$\prone_1$ to $\prone_3$; on the other hand any token coming from
$\prone_2$ gets stuck. This is because the $\gamma$ cells in the middle
work by pushing a symbol on one of the token's stacks when traversed
from their secondary to their principal port, and, dually, they pop
the symbol when traversed the other way round. The net rewrites to
another simple net consisting of an $\varepsilon$ cell connected to
the free port $\prone_2$ and a wire connecting two free ports
$\prone_1$ and $\prone_3$, which (obviously) exhibits the same
behaviour seen as a token transformer. This is indeed the way token
machines are proved to be a sound model of net
reduction in ICs~\cite{Lafont}.

The same kind of framework cannot work in MICs. To convince yourself about
that, please consider the net in Figure~\ref{fig:delta2example}, and
call it $\netone$.
\begin{figure}
  \begin{center}
  \fbox{
    \begin{minipage}{.47\textwidth}
        \centering
        \includegraphics[scale=0.7]{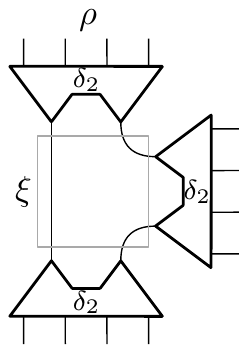}
        \caption{Three $\delta_2$ Cells Facing Each Other}\label{fig:delta2example}
  \end{minipage}}
\end{center}
\end{figure}
\begin{figure}
  \begin{center}
    \fbox{
      \begin{minipage}[c]{.47\textwidth}
        \centering
        \includegraphics[scale=0.7]{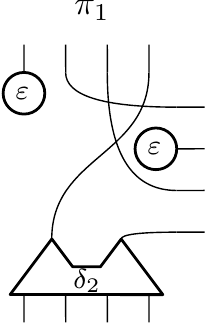}
        \qquad\quad
        \includegraphics[scale=0.7]{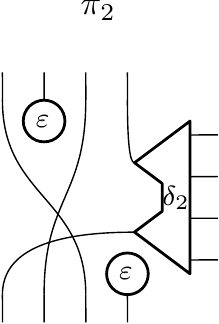}
        \qquad\quad
        \includegraphics[scale=0.7]{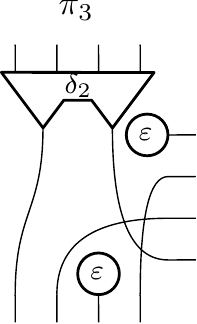}
        \caption{Three Possible Outcomes}\label{fig:delta2examplereduct}
      \end{minipage}}
  \end{center}
\end{figure}
It can reduce to three essentially different nets, namely the ones in
Figure~\ref{fig:delta2examplereduct}, call them
$\nettwo_1,\nettwo_2,\nettwo_3$. By analogy with the previous example,
it is clear that the way tokens travel inside $\rho$ should ``mimick''
the way they travel in \emph{each} of the $\nettwo_i$. But each of the
three $\nettwo_i$ behaves very differently from the others, i.e., the
\emph{nondeterminism} in the choice of which two of the three
$\delta_2$ cells in $\netone$ should interact is of a genuine kind. On
top of that, there is also the fact that the subnet $\xi$ of $\netone$
could be replaced by a much more complicated net (e.g. that in
Figure~\ref{fig:netexample}) and thus nondeterministic choices cannot
be (all) resolved initially, when the whole computation starts. Some
of them need to wait for other choices to be made, have to be
\emph{delayed}, and thus naturally become \emph{nonlocal}. One way to deal
with nonlocal nondeterminism is to charge tokens the task of resolving
it. The inevitable price to pay, however, is the fact that \emph{more
  than one token} may flow around the net at the same time: there is
no way to decide which ones of the many $\delta_2$ cells will be the
first ones to meet, and the way out is to be ``lazy'' and allow
\emph{all} $\delta_2$ cells to look for a partner. In $\netone$, as an
example, we would have six tokens, each starting at a principal port
of one of the three $\delta_2$ cells. Their task is to look for
another cell with which the $\delta_2$ cell they started from could
\emph{marry}. Once this is identified, the actual \emph{marriage} can
happen, but this must be an indivisible and irreversible operation,
after which the same cell cannot marry anyone else. This, in turn,
requires the presence of another kind of token, which is static and
whose role is to keep track of the ``civil state'' of any
$\delta_2$ cell, i.e. whether it is married or not.

The ones just described are the basic ingredients of the geometry
of interaction model we are going to introduce in Section~\ref{sect:MTM}
below. Before doing that, however, we need to setup a simple theory
of labelled bisimilarity without which one would not even have
a way to \emph{define} whether our semantics is correct.
\section{Locative Transition Systems and Bisimilarity}
\newcommand{\asone}{A}
\newcommand{\astwo}{B}
\newcommand{\actone}{a}
\newcommand{\acttwo}{b}
\newcommand{\dual}[1]{\mathrm{dual}_{#1}}
\newcommand{\portho}[3]{#2\bot_{#1} #3}
\newcommand{\ortho}[2]{#1\bot #2}
\newcommand{\sortho}[1]{\bot_{#1}}
\newcommand{\ltsone}{\mathbf{A}}
\newcommand{\ltstwo}{\mathbf{B}}
\newcommand{\ltsthree}{\mathbf{C}}
\newcommand{\fmaone}{\mathsf{a}}
\newcommand{\fmatwo}{\mathsf{b}}
\newcommand{\fmathree}{\mathsf{c}}
\newcommand{\fmafour}{\mathsf{d}}
\newcommand{\ssone}{S}
\newcommand{\sstwo}{T}
\newcommand{\trone}{\rightarrow}
\newcommand{\trtwo}{\Rightarrow}
\newcommand{\trthree}{\leadsto}
\newcommand{\isone}{s}
\newcommand{\istwo}{t}
\newcommand{\sone}{u}
\newcommand{\stwo}{t}
\newcommand{\sthree}{v}
\newcommand{\sfour}{w}
\newcommand{\tr}[4]{#2\stackrel{#3}{#1}#4}
\newcommand{\relone}{\mathcal{R}}

Usual, single-token machines, can be seen as token-transformers, i.e.,
as partial functions capturing the input-output behaviour of a net
seen as an automaton turning any token flowing \emph{into} the net
into a token flowing \emph{out} of it. This simple picture does not
hold anymore in the kind of multi-token machines we work with in this
paper. The causal dependencies between tokens are much more
complicated here, and we cannot simply proceed by letting token
machines to take in input \emph{all} tokens (like in
\cite{lics2014,lics2015}). It is thus natural to see the token machines
as labelled transition systems which interact with their environment
by letting tokens to flow in and out of them.

It is convenient to introduce a special class of labelled
transition systems, called \emph{locative transition systems}, where
labels are of a peculiar kind, namely consist of \emph{multisets} of actions
played at certain locations. Allowing more than one action to be
played together, atomically, in turn allows more than one token to
flow in and out at the same time. Locative transition systems support
a natural form of parallel composition, and notions of (bi)similarity
for which parallel composition can be shown to be a congruence. We
will give the status of a locative transition system to multi-token
machines in Section \ref{sect:MTM} below.
\subsection{Locative Transition Systems}\label{sect:lts}
\emph{Action sets}, i.e., finite set of actions, are ranged over by
metavariables like $\asone$ and $\astwo$.  For each action set
$\asone$, we assume the existence of an involution
$\dual{\asone}:\asone\rightarrow\asone$. If for some
$\actone,\acttwo\in\asone$, $\dual{\asone}(\actone)=\acttwo$, we often
write $\portho{\asone}{\actone}{\acttwo}$ or simply
$\ortho{\actone}{\acttwo}$. For every set $\setone$, let
$\msts{\setone}$ be the collection of all multisets on $\setone$,
i.e., of all functions with domain $\setone$ and codomain $\NN$.
$\fmsts{\setone}$ is the subset of $\msts{\setone}$ consisting of
\emph{finite} multisets only. To distinguish between sets and
multisets, the latter are denoted with expressions like
$\mst{\elmone_1,\dots,\elmone_n}$ instead of
$\{\elmone_1,\ldots,\elmone_n\}$, the latter standing for a set,
as usual. Let
$\lsone$ be a finite set of labels, and let $\asone$ be a set of
actions. The expression $\lsone_\asone$ stands for the collection
$\fmsts{\lsone\times\asone}$.  Elements of $\lsone_\asone$ are denoted
as $\fmaone,\fmatwo$, etc. Further, if $\pinjone$ is a partial
injection on $\lsone$, a multiset in $\fmsts{\lsone\times\asone}$ is
said to be \emph{$\pinjone$-dual} if it is in the form
$$
\mst{
  (\lblone_1,\actone_1),\ldots,(\lblone_n,\actone_n),
  (\pinjone(\lblone_1),\acttwo_1),\ldots,(\pinjone(\lblone_n),\acttwo_n)
}
$$
where $\acttwo_i=\dual{\asone}(\actone_i)$ for every $1\leq i\leq n$.
The set of $\pinjone$-dual multisets is $\sortho{\pinjone}$.

We are here interested in labelled transition systems whose labels
are drawn from a set in the form $\lsone_\asone$.  This is said to be
a \emph{locative transition system on $\lsone_\asone$} (or a
$\lsone_\asone$-LTS) and is an ordinary LTS
$\ltsone=(\ssone_\ltsone,\trone_\ltsone,\isone_\ltsone)$ where
$\ssone_\ltsone$ is a set of \emph{states},
$\isone_\ltsone\in\ssone_\ltsone$ is the \emph{initial state}, and
$\trone_\ltsone\subseteq\ssone_\ltsone\times\lsone_\asone\times\ssone_\ltsone$
is the \emph{transition relation}.
We write $\tr{\trone_\ltsone}{\sone}{\fmaone}{\stwo}$
for ${(\sone, \fmaone, \stwo)} \in {\trone_\ltsone}$.

The parallel composition of two
locative transition systems requires a slightly complicated machinery
to be defined. Given an $\lsone_\asone$-LTS $\ltsone$, an
$\lstwo_\asone$-LTS $\ltstwo$ (where $\lsone$ and $\lstwo$ are
disjoint), and a partial injection $\pinjone:\lsone\rightarrow\lstwo$,
one can define the \emph{$\pinjone$-parallel composition}
$\pc{\ltsone}{\pinjone}{\ltstwo}$ of $\ltsone$ and $\ltstwo$ as the
$\lsthree_\asone$-LTS
$(\ssone_\ltsone\times\ssone_\ltstwo,\trthree,(\isone_\ltsone,\isone_\ltstwo))$
where $\lsthree=\lsone\cup\lstwo-\dom{\pinjone}-\rng{\pinjone}$, the
transition relation $\trthree$ is
$\trthree_\ltsone\cup\trthree_\ltstwo\cup\trthree_{\ltsone\ltstwo}$,
where
\begin{align*}
  \trthree_\ltsone=&\{((\sone_\ltsone,\sone_\ltstwo),\fmaone,(\stwo_\ltsone,\sone_\ltstwo))\mid
  \fmaone\in(\lsone-\dom{\pinjone})_\asone\land\tr{\trone_\ltsone}{\sone_\ltsone}{\fmaone}{\stwo_\ltsone}\};\\
  \trthree_\ltstwo=&\{((\sone_\ltsone,\sone_\ltstwo),\fmaone,(\sone_\ltsone,\stwo_\ltstwo))\mid
  \fmaone\in(\lstwo-\rng{\pinjone})_\asone\land\tr{\trone_\ltstwo}{\sone_\ltstwo}{\fmaone}{\stwo_\ltstwo}\};\\
  \trthree_{\ltsone\ltstwo}=&\{((\sone_\ltsone,\sone_\ltstwo),\fmaone,(\stwo_\ltsone,\stwo_\ltstwo))\mid
  \exists\fmafour\in\sortho{\pinjone}.\tr{\trone_\ltsone}{\sone_\ltsone}{\fmatwo}{\stwo_\ltsone}
  \land\tr{\trone_\ltstwo}{\sone_\ltstwo}{\fmathree}{\stwo_\ltstwo}\\
  &\land\fmaone\cup\fmafour=\fmatwo\cup\fmathree
  \}.
\end{align*}
Informally, $\pc{\ltsone}{\pinjone}{\ltstwo}$ can perform an action
$\fmaone$ iff either one of $\ltsone$ or $\ltstwo$ can perform it
\emph{on one of the non-shared locations}, or both $\ltsone$ and
$\ltstwo$ perform some actions which, \emph{when put in parallel}, result
precisely in $\fmaone$.

Again, please observe that a locative transition system evolves
not by performing one (located) action $(\ell,\actone)$, but
a multiset of such actions.
\subsection{Bisimilarity}
The notion of (strong) simulation and bisimulation relations on a $\lsone_\asone$-LTS
can be defined as usual:
  \begin{varitemize}
  \item
    given two such LTSs $\ltsone$ and $\ltstwo$,
    a \emph{simulation} is any subset $\relone$ of $\ssone_\ltsone\times\ssone_\ltstwo$ such
    that 
    \begin{varnumlist}
    \item
      $\isone_\ltsone\;\relone\;\isone_\ltstwo$;
    \item
      whenever $\sone\;\relone\;\stwo$ and
      $\tr{\trone}{\sone}{\fmaone}{\sthree}$, it holds that
      $\tr{\trone}{\stwo}{\fmaone}{\sfour}$, where
      $\sthree\;\relone\;\sfour$.
    \end{varnumlist}
  \item
    given two such LTSs $\ltsone$ and $\ltstwo$, a \emph{bisimulation}
    is any subset $\relone$ of $\ssone_\ltsone\times\ssone_\ltstwo$
    such that 
    \begin{varnumlist}
    \item
      $\isone_\ltsone\;\relone\;\isone_\ltstwo$;
    \item
      whenever $\sone\;\relone\;\stwo$ and
      $\tr{\trone}{\sone}{\fmaone}{\sthree}$, it holds that
      $\tr{\trone}{\stwo}{\fmaone}{\sfour}$, where
      $\sthree\;\relone\;\sfour$;
    \item
      whenever $\sone\;\relone\;\stwo$ and
      $\tr{\trone}{\stwo}{\fmaone}{\sfour}$, it holds that
      $\tr{\trone}{\sone}{\fmaone}{\sthree}$, where
      $\sthree\;\relone\;\sfour$.
    \end{varnumlist}
  \end{varitemize}

  Similarity and bisimilarity are themselves defined as usual:
  \begin{align*}
    {\sim} &=\bigcup\{\relone\mid\mbox{$\relone$ is a bisimulation}\};\\
    {\precsim} &=\bigcup\{\relone\mid\mbox{$\relone$ is a simulation}\}.
  \end{align*}
  Weak simulation and bisimulation relations, similarity and bisimilarity
  on $\lsone_\asone$-LTSs are also defined as usual.
  Given any $\lsone_\asone$-LTS, the \emph{weak transition relation}
  $\trtwo$ is defined as
    \[
    \stackrel{\fmaone}{\trtwo}\;=\begin{cases}
    \stackrel{\emptyset}{\trone}^\ast\stackrel{\fmaone}{\trone}
    \stackrel{\emptyset}{\trone}^\ast & \text{if }\fmaone\neq\emptyset  \\
    \stackrel{\emptyset}{\trone}^\ast & \text{if }\fmaone=\emptyset.
    \end{cases}
    \]
    \begin{varitemize}
    \item
      given two such LTSs $\ltsone$ and $\ltstwo$,
      a \emph{weak simulation} is any subset $\relone$ of $\ssone_\ltsone\times\ssone_\ltstwo$ such
      that 
      \begin{varnumlist}
      \item
        $\isone_\ltsone\;\relone\;\isone_\ltstwo$;
      \item
        whenever $\sone\;\relone\;\stwo$ and
        $\tr{\trone}{\sone}{\fmaone}{\sthree}$, it holds that
        $\tr{\trtwo}{\stwo}{\fmaone}{\sfour}$, where
        $\sthree\;\relone\;\sfour$.
      \end{varnumlist}
    \item
      given two such LTSs $\ltsone$ and $\ltstwo$, a \emph{weak bisimulation}
      is any subset $\relone$ of $\ssone_\ltsone\times\ssone_\ltstwo$
      such that 
      \begin{varnumlist}
      \item
        $\isone_\ltsone\;\relone\;\isone_\ltstwo$;
      \item
        whenever $\sone\;\relone\;\stwo$ and
        $\tr{\trone}{\sone}{\fmaone}{\sthree}$, it holds that
        $\tr{\trtwo}{\stwo}{\fmaone}{\sfour}$, where
        $\sthree\;\relone\;\sfour$;
      \item
        whenever $\sone\;\relone\;\stwo$ and
        $\tr{\trone}{\stwo}{\fmaone}{\sfour}$, it holds that
        $\tr{\trtwo}{\sone}{\fmaone}{\sthree}$, where
        $\sthree\;\relone\;\sfour$.
      \end{varnumlist}
    \end{varitemize}
  Weak similarity and bisimilarity are defined as:
  \begin{align*}
    {\approx}&=\bigcup\{\relone\mid\mbox{$\relone$ is a weak bisimulation}\};\\
    {\precapprox}&=\bigcup\{\relone\mid\mbox{$\relone$ is a weak simulation}\}.
  \end{align*}
  Noticeably, parallel composition is commutative, modulo strong bisimilarity.

  \begin{lemma}
    $\pc{\ltsone}{\pinjone}{\ltstwo}\sim\pc{\ltstwo}{\pinjone^{-1}}{\ltsone}$. 
  \end{lemma}
    \begin{proof}
      We show that
      $$
      \relone=\{((\sone_\ltsone,\sone_\ltstwo),(\sone_\ltstwo,\sone_\ltsone))\mid
      \sone_\ltsone\in\ssone_\ltsone,\sone_\ltstwo\in\ssone_\ltstwo\}
      $$
      is a bisimulation. We have $(\isone_\ltsone,\isone_\ltstwo)\;\relone\;(\isone_\ltstwo,\isone_\ltsone)$.
      Assume $(\sone_\ltsone,\sone_\ltstwo)\;\relone\;(\sone_\ltstwo,\sone_\ltsone)$ and
      $\tr{\trone}{(\sone_\ltsone,\sone_\ltstwo)}{\fmaone}{(\stwo_\ltsone,\stwo_\ltstwo)}$
      in $\pc{\ltsone}{\pinjone}{\ltstwo}$. Let us distinguish three cases:
      \begin{varitemize}
      \item
        If $((\sone_\ltsone,\sone_\ltstwo),\fmaone,(\stwo_\ltsone,\stwo_\ltstwo))\in\;\trthree_\ltsone$,
        then $\sone_\ltstwo=\stwo_\ltstwo$ and $\tr{\trone}{\sone_\ltsone}{\fmaone}{\stwo_\ltsone}$.
        One can see that $\pc{\ltstwo}{\pinjone}{\ltsone}$ is able to perform
        $\tr{\trone}{(\sone_\ltstwo,\sone_\ltsone)}{\fmaone}{(\sone_\ltstwo,\stwo_\ltsone)}$
        since $((\sone_\ltstwo,\sone_\ltsone),\fmaone,(\sone_\ltstwo,\stwo_\ltsone))\in\;\trthree_\ltsone$.
      \item
        If $((\sone_\ltsone,\sone_\ltstwo),\fmaone,(\stwo_\ltsone,\stwo_\ltstwo))\in\;\trthree_\ltstwo$,
        then we can proceed in the similar way.
      \item
        If $((\sone_\ltsone,\sone_\ltstwo),\fmaone,(\stwo_\ltsone,\stwo_\ltstwo))\in\;\trthree_{\ltsone\ltstwo}$,
        then there exist $\fmatwo\in\lsone_\asone$, $\fmathree\in\lstwo_\asone$ and $\fmafour\in\sortho{\pinjone}$
        such that $\tr{\trone}{\sone_\ltsone}{\fmatwo}{\stwo_\ltsone}\wedge
        \tr{\trone}{\sone_\ltstwo}{\fmathree}{\stwo_\ltstwo}\wedge\fmaone\cup\fmafour=\fmatwo\cup\fmathree$.
        By definition, $\pinjone$-dual is equivalent to $\pinjone^{-1}$-dual, thus we indeed have
        $((\sone_\ltstwo,\sone_\ltsone),\fmaone,(\stwo_\ltstwo,\stwo_\ltsone))\in\;\trthree_{\ltstwo\ltsone}$.
      \end{varitemize}
      The opposite direction is similar.
    \end{proof}
  As expected, strong bisimilarity is a congruence for parallel composition:
  \begin{proposition}\label{prop:bisimcong}
    If $\ltsone\sim\ltstwo$, then $\pc{\ltsone}{\pinjone}{\ltsthree}\sim\pc{\ltstwo}{\pinjone}{\ltsthree}$.
  \end{proposition}
  \begin{proof}
    Let $\relone$ be a bisimulation relation between $\ltsone$ and $\ltstwo$. We show that
    $$
    \hat{\relone}=\{((\sone_\ltsone,\sone_\ltsthree),(\sone_\ltstwo,\sone_\ltsthree))\mid
    (\sone_\ltsone,\sone_\ltstwo)\in\relone,\sone_\ltsthree\in\ssone_\ltsthree\}
    $$
    is a bisimulation. We have $(\isone_\ltsone,\isone_\ltsthree)\;\hat{\relone}
    \;(\isone_\ltstwo,\isone_\ltsthree)$ since $\isone_\ltsone\;\relone\;\isone_\ltstwo$.
    Assume that $(\sone_\ltsone,\sone_\ltsthree)\;\hat{\relone}\;(\sone_\ltstwo,\sone_\ltsthree)$,
    and $\tr{\trone}{(\sone_\ltsone,\sone_\ltsthree)}{\fmaone}{(\stwo_\ltsone,\stwo_\ltsthree)}$.
    \begin{varitemize}
    \item
      If $((\sone_\ltsone,\sone_\ltsthree),\fmaone,(\stwo_\ltsone,\stwo_\ltsthree))\in\;\trthree_\ltsone$,
      then $\sone_\ltsthree=\stwo_\ltsthree$ and $\tr{\trone}{\sone_\ltsone}{\fmaone}{\stwo_\ltsone}$.
      The latter and $\sone_\ltsone\;\relone\;\sone_\ltstwo$ imply
      $\tr{\trone}{\sone_\ltstwo}{\fmaone}{\stwo_\ltstwo}$ and $\stwo_\ltsone\;\relone\;\stwo_\ltstwo$
      for some $\stwo_\ltstwo$. Therefore $\pc{\ltstwo}{\pinjone}{\ltsthree}$ have the corresponding transition
      $((\sone_\ltstwo,\sone_\ltsthree),\fmaone,(\stwo_\ltstwo,\stwo_\ltsthree))\in\;\trthree_\ltsone$.
    \item
      If $((\sone_\ltsone,\sone_\ltsthree),\fmaone,(\stwo_\ltsone,\stwo_\ltsthree))\in\;\trthree_\ltsthree$,
      then $\sone_\ltsone=\stwo_\ltsone$ and $\tr{\trone}{\sone_\ltsthree}{\fmaone}{\stwo_\ltsthree}$.
      From the latter, it immediately follows that
      $\tr{\trone}{(\sone_\ltstwo,\sone_\ltsthree)}{\fmaone}{(\sone_\ltstwo,\stwo_\ltsthree)}$.
    \item
      If $((\sone_\ltsone,\sone_\ltsthree),\fmaone,(\stwo_\ltsone,\stwo_\ltsthree))\in\;\trthree_{\ltsone\ltsthree}$,
      then there exist $\fmatwo\in\lsone_\asone$, $\fmathree\in\lstwo_\asone$ and $\fmafour\in\sortho{\pinjone}$
      such that $\tr{\trone}{\sone_\ltsone}{\fmatwo}{\stwo_\ltsone}\wedge
      \tr{\trone}{\sone_\ltsthree}{\fmathree}{\stwo_\ltsthree}\wedge\fmaone\cup\fmafour=\fmatwo\cup\fmathree$.
      We use $\sone_\ltsone\;\relone\;\sone_\ltstwo$ to obtain $\stwo_\ltstwo$ such that
      $\tr{\trone}{\sone_\ltstwo}{\fmaone}{\stwo_\ltstwo}$ and $\stwo_\ltsone\;\relone\;\stwo_\ltstwo$.
      Now observe
      $((\sone_\ltstwo,\sone_\ltsthree),\fmaone,(\stwo_\ltstwo,\stwo_\ltsthree))\in\;\trthree_{\ltstwo\ltsthree}$.
    \end{varitemize}
    This concludes the proof.
  \end{proof}
  The same holds for weak bisimilarity:
  \begin{proposition}\label{prop:weakbisimcong}
    If $\ltsone\approx\ltstwo$, then $\pc{\ltsone}{\pinjone}{\ltsthree}\approx\pc{\ltstwo}{\pinjone}{\ltsthree}$.
  \end{proposition}
\section{Multi-Token Machines}\label{sect:MTM}
This section is devoted to defining the multi-token machines which are
the object of study of this paper, and to proving some basic properties of them.
In particular, we will give a Compositionality Theorem that will be very
helpful in the following section.
\subsection{States}
\newcommand{\lft}{\mathtt{p}}
\newcommand{\rgt}{\mathtt{q}}
\newcommand{\emptystk}{\epsilon}
\newcommand{\skone}{\mathbf{s}}
\newcommand{\sktwo}{\mathbf{t}}
\newcommand{\skthree}{\mathbf{u}}
\newcommand{\skfour}{\mathbf{v}}
\newcommand{\skset}{\mathcal{STK}}
\newcommand{\conone}{C}
\newcommand{\contwo}{D}
\newcommand{\conthree}{E}
\newcommand{\conset}{\mathcal{CON}}
\newcommand{\ctone}{c}
\newcommand{\cttwo}{d}
\newcommand{\cnj}[1]{\overline{#1}}
\newcommand{\tkns}{\mathcal{TKS}}
\newcommand{\tknone}{T}
\newcommand{\tkntwo}{X}
\newcommand{\tknthree}{Y}
\newcommand{\tknfour}{Z}
\newcommand{\stone}{\mathsf{S}}
\newcommand{\sttwo}{\mathsf{U}}
\newcommand{\stthree}{\mathsf{V}}
\newcommand{\stfour}{\mathsf{T}}

States of a multi-token machine are just multisets of tokens, each
of them consisting of a port in the underlying net, modelling \emph{where}
the token is, and of some auxiliary information (e.g. the token's
origin, some stacks, etc.), which varies depending on the token's kind.
All this will be formalised in this section.

A \emph{stack} $\skone$ is any sequence whose elements are either
symbols from the alphabet $\{\lft,\rgt\}$ or natural
numbers. Formally:
$$
\skone,\sktwo,\skthree\bnf\emptystk\midd\skone\lft\midd\skone\rgt\midd\skone\natone,
$$
where $\natone\in\NN$ and $\emptystk$ is the empty stack. The set of
all stacks is $\skset$.  A \emph{configuration} is a pair of stacks
$(\skone,\sktwo)$, and is denoted with metavariables like
$\conone,\contwo$. The set of all configurations is $\conset$.  We
define the \emph{conjugate} $\cnj{\skone}$ of a stack $\skone$ as
follows:
\begin{align*}
  \cnj{\emptystk} &= \emptystk &
  \cnj{\skone\lft} &= \cnj{\skone}\rgt &
  \cnj{\skone\rgt} &= \cnj{\skone}\lft &
  \cnj{\skone\natone} &= \cnj{\skone}\natone
\end{align*}
Given a configuration $\conone=(\skone, \sktwo)$, its conjugate
$\cnj{\conone}$ is defined to be $(\cnj{\skone}, \sktwo)$.

Throughout this section, wires of the underlying net will be ranged
over by metavariables like $\wrone,\wrtwo,\wrthree$.  A
\emph{cell type} is any element of the set $\ctset=\{\gamma, \delta\}$.
Cell types will be denoted with metavariables like $\ctone$ an $\cttwo$.
\emph{Tokens} can be of one of four different kinds:
\begin{varitemize}
\item
  \emph{Single status tokens}, which are elements of
  $\prset^\netone\times\skset$.  Graphically, single status tokens are
  denoted with $\blacksquare$.
\item
  \emph{Married status tokens}, which are elements of
  $\prset^\netone\times\skset\times\ctset$.  Graphically, married
  status tokens are denoted with $\times$.
\item
  \emph{Marriage tokens}, which are elements of
  $\prset^\netone\times\conset\times\prset^\netone\times\conset$.
  Graphically, marriage tokens are denoted with $\medbullet$.
\item
  \emph{Matching tokens}, which are elements of
  $\prset^\netone\times\conset\times\prset^\netone\times\conset$.
  Graphically, matching tokens are denoted with $\medcirc$.
\end{varitemize}
The first components of tokens indicate their current positions.
Status tokens must be placed at principal ports while matching tokens
must lie at free ports. The second components of marriage and
matching tokens are their current configurations.  At the third and
fourth components, marriage and matching tokens keep their
origins and configurations as initially installed.  The set of all tokens
for the net $\netone$ is indicated as $\tkns^\netone$.

\emph{States} of the machine we are defining will just be multisets of
tokens, and are denoted with metavariables like $\stone$ and
$\sttwo$.  We assume that marriage tokens on auxiliary or principal
ports of cells are always going out of the cells while those on free
ports are coming in from the environment. For example, we depict two
marriage tokens $(\prone, (\skone, \sktwo), \prtwo, (\skthree,
\skfour))$ on a principal port $\prone$ of a $\gamma$ cell and
$(\prone', (\skone', \sktwo'), \prtwo', (\skthree', \skfour'))$ on a
port $\prone'$ which is connected with $\prone$ by wire ($\prone'$ may
be either a port of another cell or a free port) as in
Figure~\ref{fig:graphmarriagetokens}.

The \emph{initial state} of the machine comprises one marriage token
$(\prone, (\emptystk,\emptystk),\prone, (\emptystk,\emptystk))$
and one single status token $(\prone, \emptystk)$
for each principal port $\prone$ of each $\delta_2$ cell in the underlying net,
as shown in Figure~\ref{fig:initstate}.
Intuitively, those tokens in the initial state are in charge of
keeping track of the status of the $\delta_2$ cell, and to look for
possible partners for it.

\begin{figure}
\begin{center}
\fbox{
  \begin{minipage}{.47\textwidth}
      \centering
      \includegraphics[scale=1]{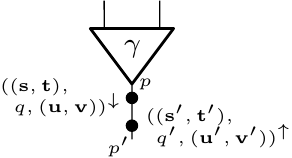}
      \caption{Marriage Tokens: an Example}\label{fig:graphmarriagetokens}
  \end{minipage}}
\end{center}
\end{figure}
\begin{figure}
\begin{center}
\fbox{
  \begin{minipage}{.47\textwidth}
      \centering
      \includegraphics[scale=1]{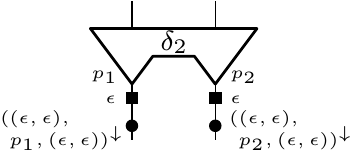}
      \caption{Initial State of Token Machine}\label{fig:initstate}
  \end{minipage}}
\end{center}
\end{figure}

\subsection{Internal Transition Rules}
The behaviour of multi-token machines is given by a series of
transition rules which prescribe how and when a marriage token can move inside
a net, and the protocol governing marriages. 

Transition rules prescribing how token \emph{moves} are in
Figure~\ref{fig:movtrans}.
\begin{figure*}
  \begin{center}
  \fbox{
    \begin{minipage}{.97\textwidth}
      \begin{center}
      \includegraphics[scale=0.8]{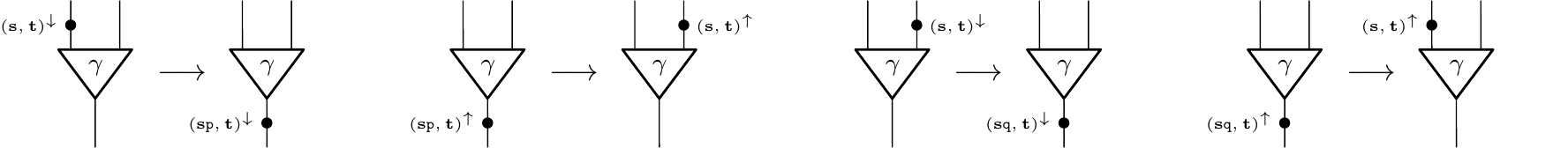}
      \end{center}
      \begin{center}
      \includegraphics[scale=0.8]{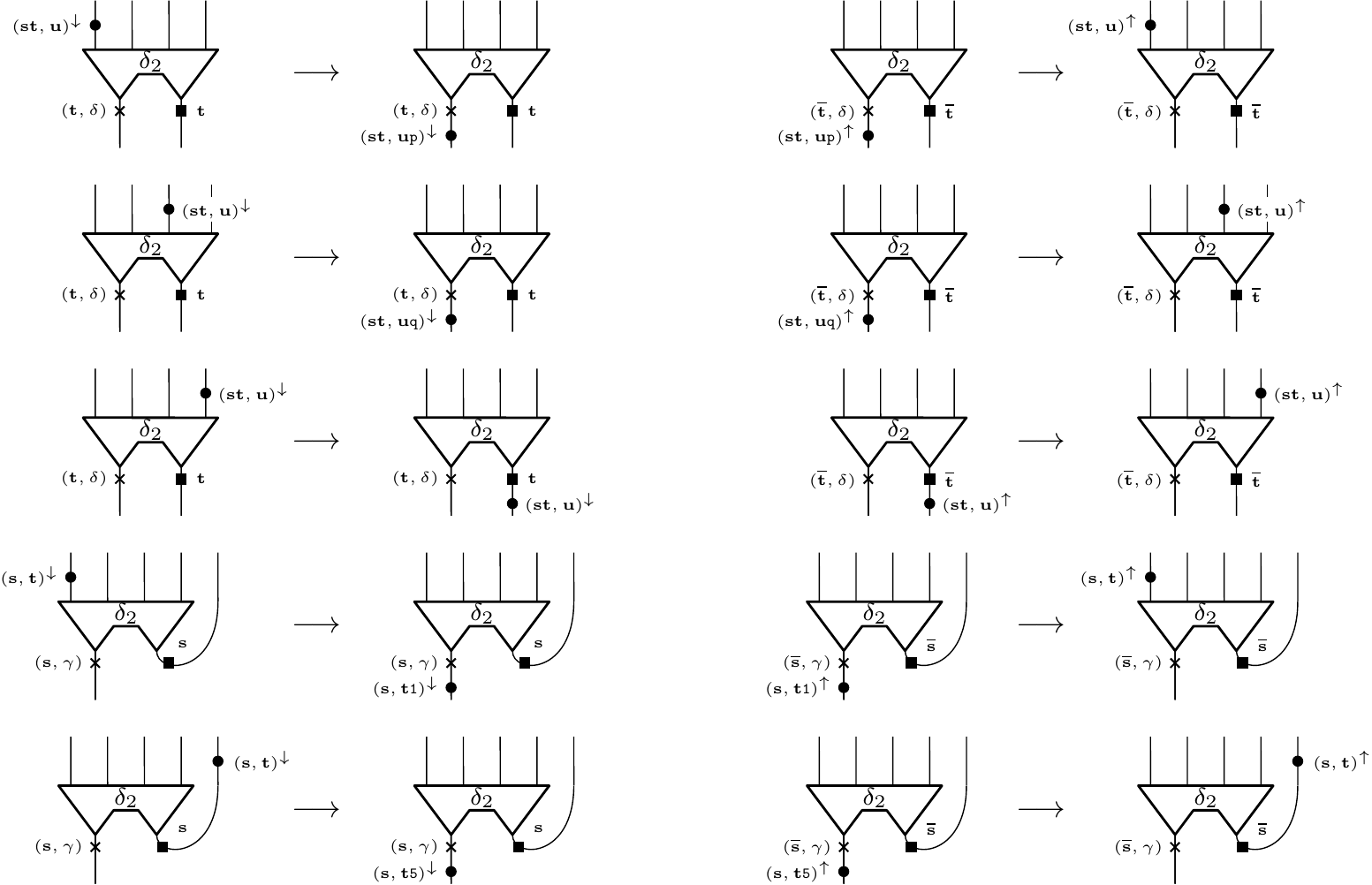}
      \end{center}
      \caption{Internal Moving Transition Rules}\label{fig:movtrans}
  \end{minipage}}
  \end{center}
\end{figure*}
Please observe how marriage tokens can flow through $\gamma$ cells
more or less the same way they do in interaction combinators~\cite{Lafont}:
some symbol is either pushed or popped from their first stack.
On the other hand, when tokens face a $\delta_2$ cell,
there are in principle more than one possibility as for how they should
move, all this depending on the presence of certain married status tokens.

As for transition rules performing \emph{marriages}, they are in
Figure~\ref{fig:martrans}.
\begin{figure*}
  \begin{center}
  \fbox{
    \begin{minipage}{.97\textwidth}
      \begin{center}
      \includegraphics[scale=0.9]{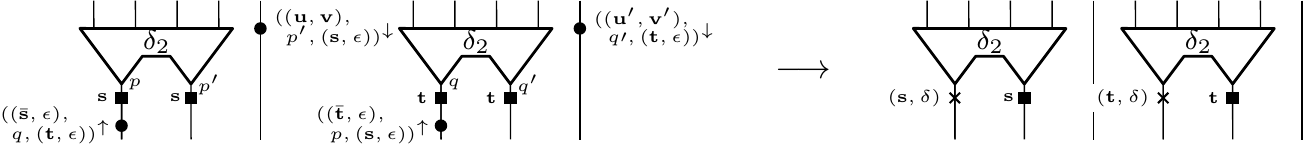}
      \end{center}
      \begin{center}
      \includegraphics[scale=0.9]{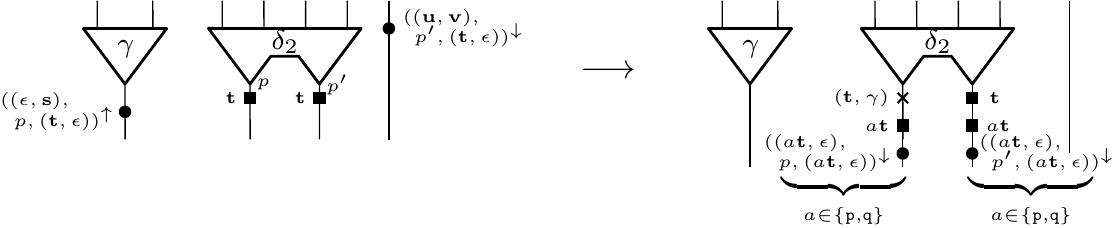}
      \end{center}
      \caption{Internal Marriage Transition Rules}\label{fig:martrans}
  \end{minipage}}
  \end{center}
\end{figure*}
The first rule handles the marriage between two $\delta_2$ cells,
while the second one is charged of handling the case in which a
$\delta_2$ cell marries a $\gamma$ cell. In the first case, it is as
if the two $\delta_2$ cells have interacted with each other, thus
annihilating themselves, and producing some $\varepsilon$ cells.
Notice also how the marriage token which originated from \emph{the
  other} principal port, is annihilated in the process; summing up,
there are four tokens involved altogether.  In the second
case, the $\delta_2$ is virtually duplicated, and indeed some new
(status and marriage) tokens are created, each one corresponding to a
virtual copy of the $\delta_2$ cell. In both cases, (single and
married) status tokens are in charge of guaranteeing atomicity, 
and of keeping track of whether each (copy of a) $\delta_2$ cell
is single or married, in the latter case remembering also the nature
of the cell's partner. It is instructive to notice how marriages
between two $\delta_2$ cells are bidirectional and symmetric in
nature, while those between a $\delta_2$ cell and a $\gamma$ cell
are asymmetric and unidirectional: the latter are inert and marriage
tokens are meant to start their journey from $\delta_2$ cells
uniquely.
\subsection{External Transition Rules}
\newcommand{\outa}[1]{\mathsf{out}(#1)}
\newcommand{\ina}[1]{\mathsf{in}(#1)}
\newcommand{\kla}[1]{\mathsf{kill}(#1)}
\newcommand{\klda}[1]{\mathsf{cokill}(#1)}
\newcommand{\maa}[2]{\mathsf{ma}_{#1}(#2)}
\newcommand{\eaone}{A}
\newcommand{\eatwo}{B}
\newcommand{\easet}{\mathcal{EA}}
\newcommand{\tkm}[1]{\mathsf{TM}_{#1}}

As already pointed out, it is quite convenient to see token machines
not simply as automata evolving as described in the previous
section, but also as labelled transition systems having a nontrivial
\emph{interactive} behaviour. More specifically, token machines can
interact with their environment through free ports by inputting or outputting a token,
by letting a marriage happen, or by killing a token as a part of a
marriage. This makes locative transition systems an ideal candidate
for the kind of LTSs one needs here. \emph{External actions}, which express
the interactions which could happen at such a location,
are generated by the following grammar:
$$
\eaone,\eatwo\bnf\outa{\conone}\midd\ina{\conone}\midd\maa{\ctone\cttwo}{\conone}
\midd\kla{\conone}\midd\klda{\conone}
$$
where $\conone$ ranges over $\conset$ and $\ctone,\cttwo$ ranges over
$\ctset$. The set of external actions is indicated as $\easet$. We define
$\dual{\easet}$ by the following rules:
$$
  \outa{\conone}\bot\ina{\conone};\qquad
  \maa{\ctone\cttwo}{\conone}\bot\maa{\cttwo\ctone}{\overline{\conone}};\qquad
  \kla{\conone}\bot\klda{\conone}.
$$ 
Actually, any net $\netone$ can be turned into an
  $\fprset^\netone_{\easet}$-LTS
  $\tkm{\netone}=(\fmsts{\tkns^\netone},\red,\isone_\netone)$, 
which is said to be the \emph{token machine for $\netone$}.  This is
done considering internal transition rules as producing the empty
multiset of action, and by giving some \emph{external} transition
rules, namely those in Figure~\ref{fig:externaltrans} and
Figure~\ref{fig:killtrans}.  Multiple external actions involving
distinct tokens can be combined in just one labelled transition: this
is possible because labels of $\tkm{\netone}$ are multisets.
\begin{figure*}
  \begin{center}
  \fbox{
    \begin{minipage}{.97\textwidth}
      \begin{center}
        \includegraphics[scale=0.9]{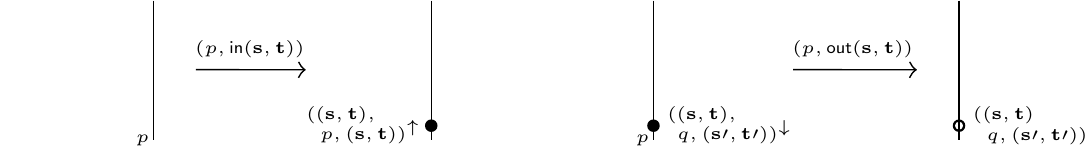}
      \end{center}
      \begin{center}
        \includegraphics[scale=0.9]{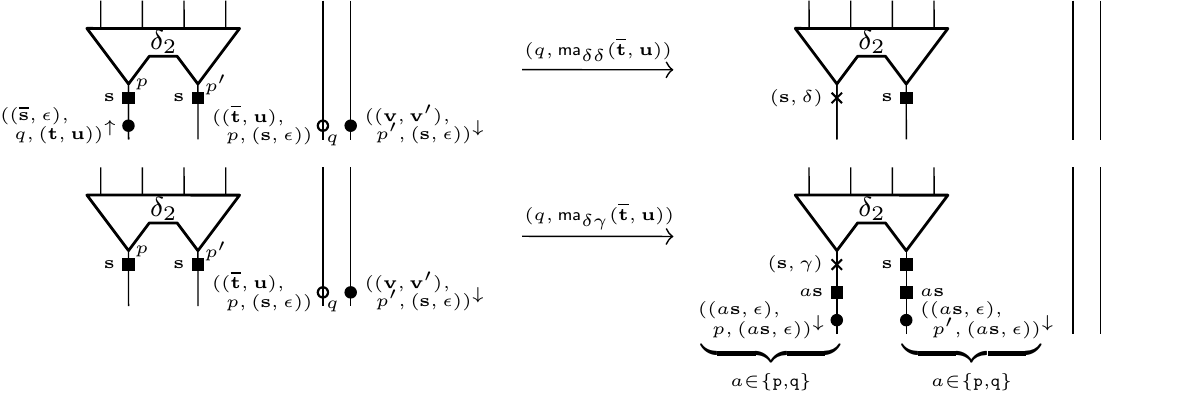}
      \end{center}
      \begin{center}
        \includegraphics[scale=0.9]{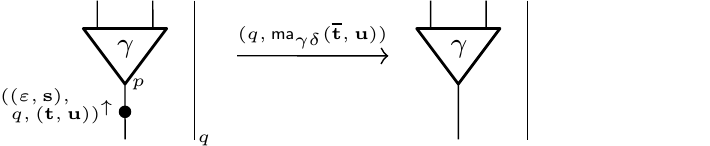}
      \end{center}
      \begin{center}
        \includegraphics[scale=0.9]{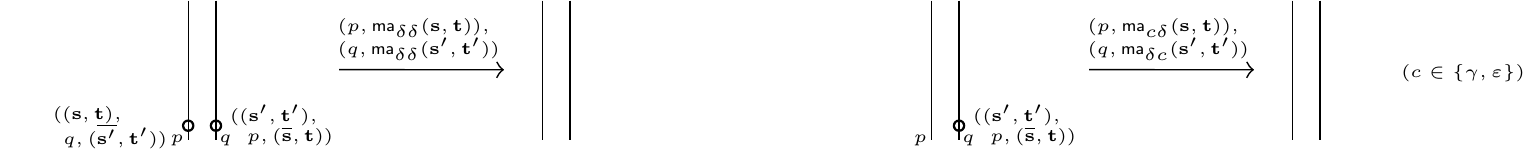}
      \end{center}
      \caption{External Transition Rules for Input, Output and Marriages}\label{fig:externaltrans}
  \end{minipage}}
  \end{center}
\end{figure*}
\begin{figure*}
  \begin{center}
  \fbox{
    \begin{minipage}{.97\textwidth}
      \begin{center}
        \includegraphics[scale=0.9]{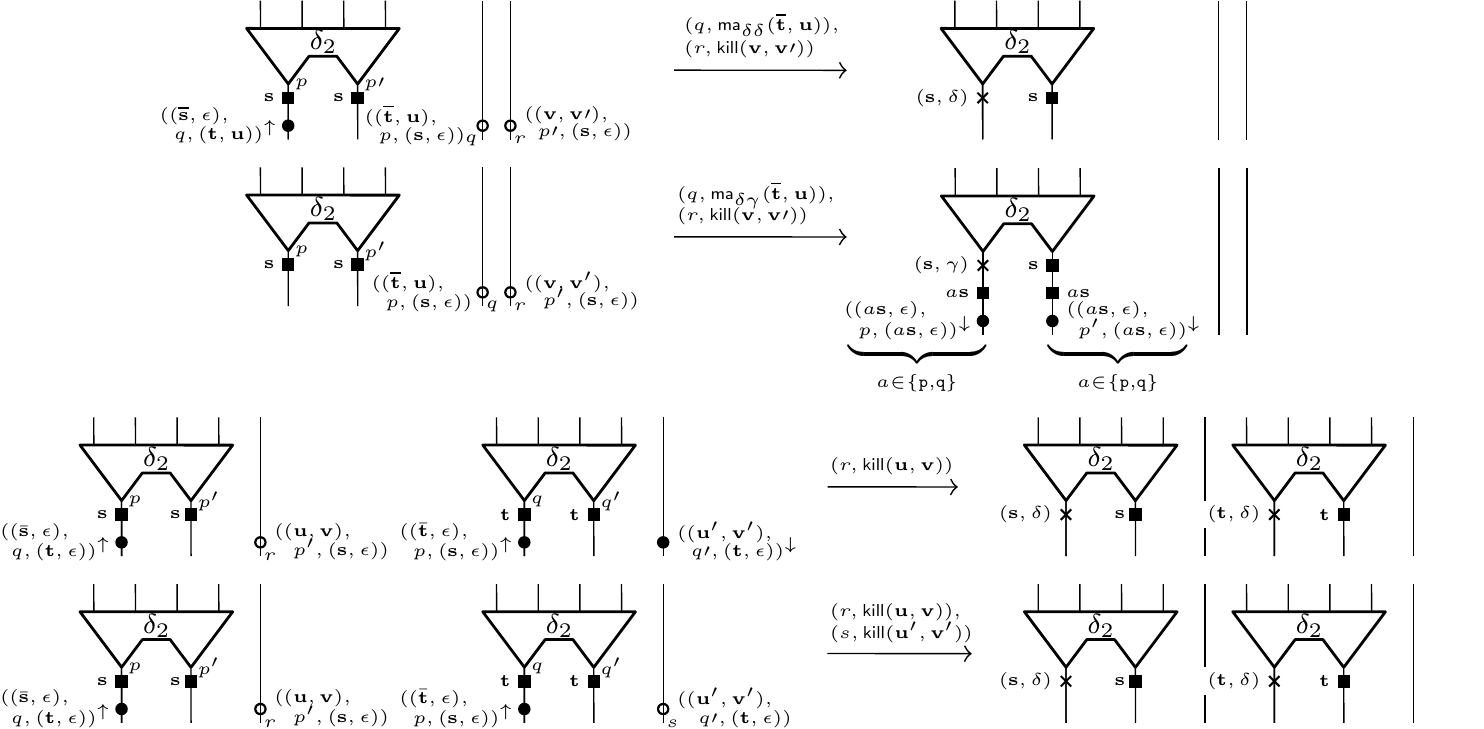}
      \end{center}
      \begin{center}
        \includegraphics[scale=0.9]{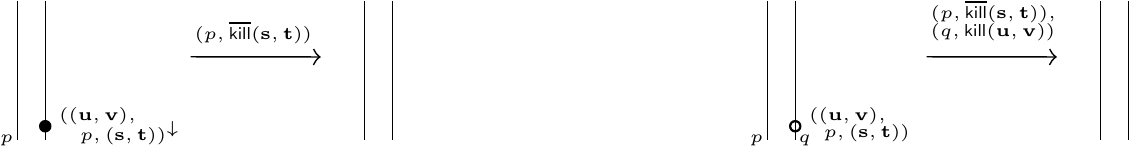}
      \end{center}
      \caption{External Transition Rules for Killing Tokens}\label{fig:killtrans}
  \end{minipage}}
  \end{center}
\end{figure*}
Let us briefly comment on the role of each rule:
\begin{varitemize}
\item
  The first two rules in Figure~\ref{fig:externaltrans}
  allow tokens to flow into the net (thus giving rise to
  an action $\ina{\conone}$) or to flow out of the net
  (thus giving rise to an action $\outa{\conone}$).
  In the latter case, something needs to keep track of the fact
  that a marriage token has indeed left the net (in a certain
  configuration): this is precisely the role of matching tokens,
  that is to say tokens of the fourth (and last) kind.
\item
  The fourth and fifth rules in Figure~\ref{fig:externaltrans}
  model the situation in which  a $\delta_2$ cell performs a
  marriage ``with the environment'', by way of a pair of
  marriage tokens, the first of which is as usual next to
  the cell, but the other has flown out the net, a fact
  witnessed by the presence of a matching token. The fourth
  rule takes care of the case in which the cell in the
  environment is itself a $\delta_2$ cell, while the
  fifth rule accounts for an asymmetric marriage with a
  $\gamma$ cell.
\item
  We also have a rule (the third one in Figure~\ref{fig:externaltrans})
  that models the marriage between a $\gamma$ cell and the
  environment.
\item
  The sixth and seventh rules in Figure~\ref{fig:externaltrans}
  forward marriage actions via the matching tokens, emitting
  dual actions.
\item
  In the third and fourth rules in Figure~\ref{fig:externaltrans}, the
  marriage token which needs to be killed as a result of the marriage
  still lies in the net, but nothing guarantees that it has not flown
  out of the net itself.  It is thus necessary to have some further
  rules, for example the first rule in
    Figure~\ref{fig:killtrans}, in which we not only
  marry a $\delta_2$ cell with the environment but also kill a
  marriage token which has already flown out of the net, expressed by
  a $\kla{\conone}$ action.
\item
  The last two rules in Figure~\ref{fig:killtrans}
  describe how a $\kla{\conone}$ action is forwarded and how it
  actually acts on a marriage token, by emitting the dual
  action $\klda{\conone}$.
\end{varitemize}
Having those external rules, two nets can now interact via external actions.
The precise situation which it yields will be analysed in the next sections.
\subsection{Basic Properties}
\newcommand{\fred}[1]{\stackrel[#1]{}{\rightarrow}}
\newcommand{\lblred}[1]{\stackrel{#1}{\rightarrow}}
\newcommand{\flblred}[2]{\stackrel[#2]{#1}{\rightarrow}}
\newcommand{\op}[1]{{#1}^{\mathit{OP}}}
\newcommand{\opworig}[2]{{#1}^{\mathit{OP}}_{#2}}
\newcommand{\fcone}{\Theta} 
\newcommand{\fctwo}{\Xi}
\newcommand{\fcthree}{\Pi}
\newcommand{\mtm}[1]{\mathsf{M2M}(#1)}
\newcommand{\orig}[1]{\mathsf{Orig}(#1)}

The expression $\lblred{\emptyset}$ where $\emptyset$ is the
empty multiset will be indicated simply as $\red$.
\begin{lemma}\label{lemma:statusinflationary}
  Marriage status tokens are inflationary, i.e., if $\stone=\sttwo\cup\{\tknone\}$,
  $\tknone$ is a status token, and $\stone\lblred{\lblone}\stthree$,
  then $\tknone\in\stthree$.
\end{lemma}
\begin{proof}
  This is simple case analysis: whenever we perform any internal or
  external action, status tokens can be created, but once they are
  part of the current state, they never change their status (and, in
  particular, they never disappear).
\end{proof}
In the rest of this section, we develop notions and lemmas to show
Theorem~\ref{thm:compositionality}, which will be proved in the next
Section \ref{sect:compositionality}. Suppose we are working with
$\tkm{\netone}$ and suppose that
$\stone\cup\sttwo\lblred{\lblone}\stthree$ where, however, only the
tokens from $\sttwo$ are allowed to evolve, while those from $\stone$
cannot change their status. If this is the case, then $\stthree$ is
necessarily in the form $\stone\cup\stfour$, and we write
$\sttwo\flblred{\lblone}{\stone}\stfour$.
Given a marriage token $\tknone = (\prone, \conone, \prtwo, \contwo)$ and
a pair $(\prthree, \conthree)$ of a port $\prthree$ and a configuration $\conthree$,
a token $\opworig{\tknone}{\prthree, \conthree}$ is defined as
$(\prone, \cnj{\conone}, \prthree, \conthree)$.

Given a token $\tknone$ and multiset of tokens $\stone$, an
\emph{$(\stone,\tknone)$-focused sequence ending in $\tknthree$} is
any sequence $\fcone$ of tokens in the form
$$
\fcone:\tkntwo_1\fred{\stone}\tkntwo_2\fred{\stone}\cdots\fred{\stone}\tkntwo_n,
$$ 
where $\tkntwo_1=\tknone$ and $\tkntwo_n=\tknthree$. The sequence
$\fcone$ are also called an \emph{$(\stone,\tknone)$-focused sequence}
or even an \emph{$\stone$-sequence}. A state $\stone$ is said to be
\emph{reachable from $\sttwo$} iff there is a sequence of labelled
transitions leading $\sttwo$ to $\stone$:
$$
\sttwo\lblred{\lblone_1}\cdots\lblred{\lblone_n}\stone.
$$
A state $\stone$ is said to be \emph{reachable} iff it is reachable
from the initial state of the underlying net. A marriage token is said
to be an \emph{origin token} iff its origin and current positions
coincide. For every marriage token
$\tknone=(\prone,\conone,\prtwo,\contwo)$, \emph{its origin token} is
$\orig{\tknone}=(\prtwo,\contwo,\prtwo,\contwo)$.  For every matching
token $\tknone$, there is a naturally defined marriage token
$\mtm{\tknone}$, namely the one that originated $\tknone$.  Given a
state $\stone$ and a marriage token $\tknone$, $\tknone$ is said to be
\emph{$\stone$-canonical} iff there is an
$(\stone,\orig{\tknone})$-focused sequence ending in $\tknone$. A
matching token $\tknone$ is \emph{$\stone$-canonical} if
$\mtm{\tknone}$ is. If $\stone$ is such that for every marriage or
matching token $\tknone$ in $\stone$, $\tknone$ is
$(\stone-\{\tknone\})$-canonical, then $\stone$ is said to be, simply,
\emph{canonical}. Canonicity is preserved by internal or external
interaction:
\begin{lemma}\label{lemma:preservecanonicity}
  If $\stone$ is canonical and
  $\stone\lblred{\lblone}\sttwo$, then
  $\sttwo$ is canonical, too.
\end{lemma}
\begin{proof}
This can be proved by a simple case analysis on the rule for
$\stone\lblred{\lblone}\sttwo$.
\end{proof}
The way one builds $\stone$-focused sequences is essentially
a deterministic process:
\begin{lemma}
  Given two $\stone$-focused sequences $\fcone$ and $\fctwo$, either
  $\fcone$ is a prefix of $\fctwo$ or vice versa.
\end{lemma}
No nondeterminism is involved in building $\stone$-focused sequences:
\begin{lemma}\label{lemma:invertibility}
  $\stone$-focused sequences are invertible, i.e.\ for every such computation
  $$
  \fcone:\tkntwo_1\fred{\stone}\tkntwo_2\fred{\stone}\cdots\fred{\stone}\tkntwo_n
  $$
  and for every pair $(\prthree, \conthree)$ of a port $\prthree$ and a configuration $\conthree$,
  there exists another $\stone$-focused sequence 
  $\opworig{\fcone}{\prthree, \conthree}$ such that
  $$
  \opworig{\fcone}{\prthree, \conthree}:
  \opworig{\tkntwo_n}{\prthree, \conthree} \fred{\stone}
  \opworig{\tkntwo_{n-1}}{\prthree, \conthree} \fred{\stone}
  \cdots \fred{\stone} \opworig{\tkntwo_1}{\prthree, \conthree}.
  $$
\end{lemma}
\subsection{Compositionality}\label{sect:compositionality}
\newcommand{\funone}{f}
\newcommand{\funtwo}{g}
Consider any marriage or matching token $\tknone$ in a state
$\stone\cup\{\tknone\}$ of $\tkm{\pc{\netone}{\pinjone}{\nettwo}}$.
If $\tknone$ is $\stone$-canonical, then we can trace back $\tknone$
(or $\mtm{\tknone}$) to an origin token $\tkntwo$ by way of an
$\stone$-focused sequence.  Lemma \ref{lemma:invertibility} tells us
that from $\tkntwo$ we can go back to $\tknone$ (or $\mtm{\tknone}$),
again by way of an $\stone$-focused sequence, that we call $\fcone$.
Intuitively, $\fcone$ is the computation that brought $\tknone$ where
it is now. We can see it as a path in
$\pc{\netone}{\pinjone}{\nettwo}$, and depict it graphically as
follows:
\begin{center}
  \includegraphics[scale=2.4]{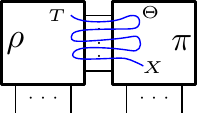}
\end{center}
In particular, $\fcone$ can cross the border between $\netone$ and
$\nettwo$ several times. If we want to somehow simulate $\fcone$ in
$\pc{\tkm{\netone}}{\pinjone}{\tkm{\nettwo}}$, we need to keep
track of each crossing by a matching token lying at the border, on
the side of the sub-net we are leaving. Let $\sttwo$ the set of all
these matching tokens which are tokens of either $\tkm{\netone}$ or of
$\tkm{\nettwo}$.  What we have just implicitly defined is a pair of
states $\funtwo_\netone(\stone,\tknone)$ and
$\funtwo_\nettwo(\stone,\tknone)$: they are defined as the portions of
$\stone\cup\sttwo$ which are tokens of $\tkm{\netone}$ and $\tkm{\nettwo}$,
respectively. If $\tknone$ is not $\stone$-canonical,
then
\begin{theorem}[Compositionality]\label{thm:compositionality}
$\tkm{\pc{\netone}{\pinjone}{\nettwo}}\approx\pc{(\tkm{\netone})}{\pinjone}{(\tkm{\nettwo})}$. 
\end{theorem}
\begin{proof}
  We need to define a binary relation $\relone$ between the states
  of the two involved LTSs, and then prove it to be a bisimulation.
  The simplest way to describe $\relone$ is as the set of all the
  pairs in the form $(\stone,\funone^*(\stone))$, where
  $\stone$ is a canonical state of $\pc{\netone}{\pinjone}{\nettwo}$,
  and $\funone^*$ is defined as
  $$
  \funone^*(\stone)=(\cup_{\tknone\in\stone}\funone_\netone(\stone,\tknone),\cup_{\tknone\in\stone}\funone_\nettwo(\stone,\tknone))
  $$
  and the partial function $\funone_\netone$ (respectively,
  $\funone_\nettwo$) maps a pair in the form $(\stone,\tknone)$ into
  a state of $\tkm{\netone}$ (respectively, of $\tkm{\nettwo}$).  The
  way $\funone_\netone$ and $\funone_\nettwo$ are defined depends on
  the nature of $\tknone$. In particular,
  \begin{varitemize}
  \item
    If $\tknone$ is either a single status token or a marriage status token
    lying next to its natural position, say a cell in $\netone$, then
    $\funone_\netone(\stone,\tknone)$ will be 
    just $\{\tkntwo\}$, where $\tkntwo$ is the counterpart of $\tknone$
    in $\netone$, while $\funone_\nettwo(\stone,\tknone)$ will be just
    $\emptyset$. Similarly when the position of $\tknone$ is a cell in
    $\nettwo$. If $\tknone$ is not at a natural position, then
    both $\funone_\netone$ an $\funone_\nettwo$ are undefined at $(\stone,\tknone)$.
  \item
    If $\tknone$ is a marriage or matching token, then
    $\funone_\netone(\stone,\tknone)=\funtwo_\netone(\stone,\tknone)$
    and $\funone_\nettwo(\stone,\tknone)=\funtwo_\nettwo(\stone,\tknone)$
    where $\funtwo_\netone$ and $\funtwo_\nettwo$ are the maps defined above.
  \end{varitemize}
  We now need to prove that $\relone$ is a bisimulation relation. To
  do that, we need to prove three statements:
  \begin{varnumlist}
  \item
    We first of all need to show that
    $(\isone_{\pc{\netone}{\pinjone}{\nettwo}},(\isone_\netone,\isone_\nettwo))\in\relone$.
    To prove that, first of all observe that
    $\isone_{\pc{\netone}{\pinjone}{\nettwo}}$ does not contain any
    matching token, that all the marriage tokens in it are trivially
    canonical, and are thus mapped simply to ``themselves'' by $\funone^*$.
  \item
    We then need to show that if $(\stone,\sttwo)\in\relone$
    and $\tr{\trone}{\stone}{\lblone}{\stthree}$, then
    $\tr{\trtwo}{\sttwo}{\lblone}{\stfour}$ where
    $(\stthree,\stfour)\in\relone$. We need to analyse
    several cases:
    \begin{varitemize}
    \item
      If $\stthree$ is obtained from $\stone$ by performing
      a crossing internal transition, then the crossing can
      be easily simulated in $\sttwo$, the only proviso being
      that if the involved marriage token in $\stone$ lies
      at the border between $\netone$ and $\nettwo$, then
      the simulating transitions could be \emph{two} instead
      of \emph{one}, namely a crossing and matching
      input-output pair.
    \item
      If $\stthree$ is obtained from $\stone$ by performing
      an internal $\delta\delta$ marriage, then the four
      involved marriage tokens are as follows
      \begin{center}
        \includegraphics[scale=1.2]{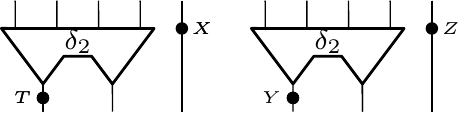}
      \end{center}
      where $\tkntwo$ and $\tknthree$ originated in the
      leftmost cell, while $\tknone$ and $\tknfour$
      originated in the rightmost cell. The four
      tokens are canonical, and are thus mapped by $\funone^*$
      to four marriage tokens, possibly with some matching
      tokens lying along the paths leading each of them
      to its origin. Now, suppose, just as an example, that the
      leftmost cell and link above are in $\netone$, while rightmost
      cell and link are in $\nettwo$ (the other cases
      can be handled similarly). Summing up, then, we
      are in the following situation:
      \begin{center}
        \includegraphics[scale=1.2]{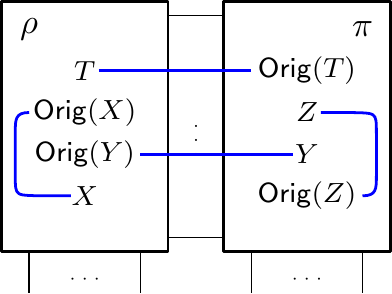}
      \end{center}
      where the four depicted paths can of course, in principle,
      go back and forth between $\netone$ and $\nettwo$, even
      several times. Now, observe that:
      \begin{varitemize}
      \item
        Along the four paths, appropriate matching tokens in
        $\funone^*(\stone)$ mark the crossing of the border between
        $\netone$ and $\nettwo$.
      \item
        Let $\tknone = (\prtwo, \cnj{\conone}, \prone, \contwo)$ and
        $\tknthree = (\prone, \cnj{\contwo}, \prtwo, \conone)$
        (by definition of the initial state and transition rules
        they are necessarily in this form ).
        $\opworig{\tknone}{\prtwo,\conone} = \orig{\tknthree}$
        and $\opworig{\tknthree}{\prone,\contwo} = \orig{\tknone}$,
        and 
        by Lemma~\ref{lemma:invertibility}, there exist paths
        from $\orig{\tknone}$ to $\tknone$ and to $\orig{\tknthree}$
        to $\tknthree$ that are the inverses of each other. Then, the
        two paths and the matching tokens along them are of the form:
        \begin{center}
          \includegraphics[scale=1.2]{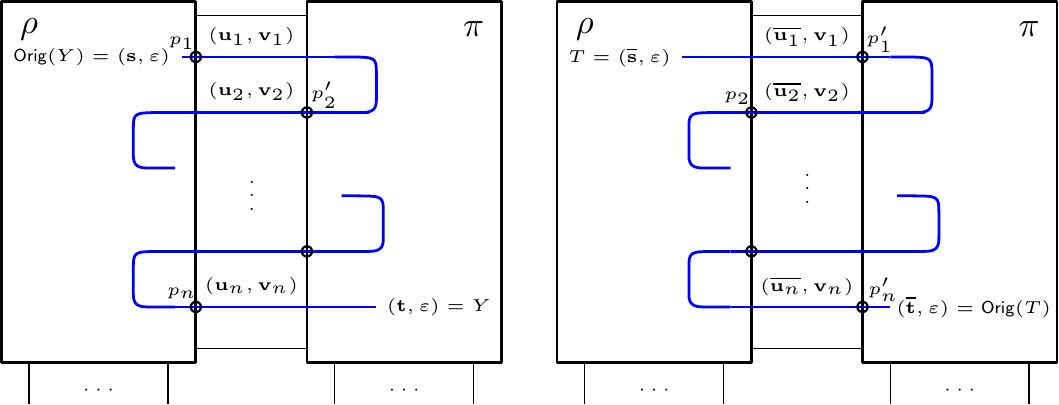}
        \end{center}
        where $\pinjone(\prone_i)=\prone'_i$ and $n=2n'+1$ for some $n'\geq0$.
      \item
        The paths from $\orig{\tkntwo}$ to $\tkntwo$ and from $\orig{\tknfour}$
        to $\tknfour$ are depicted together with the matching tokens
        along them as follows:
        \begin{center}
          \includegraphics[scale=1.2]{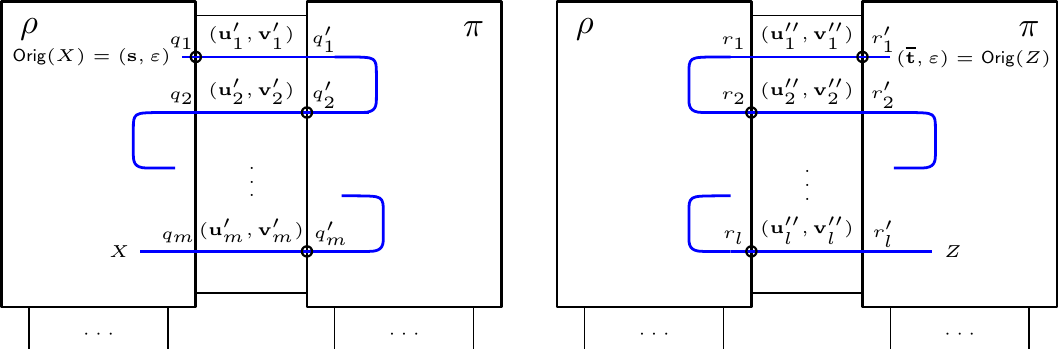}
        \end{center}
        where $\pinjone(\prtwo_i)=\prtwo'_i$, $\pinjone(\prthree_i)=\prthree'_i$,
        $m=2m'$ and $l=2l'$ for $m',l'\geq0$.
      \end{varitemize}
      As a consequence, one can prove that $\sttwo$ silently
      evolves to a state $\stfour$. More specifically, if
      $\sttwo=(\sttwo_\netone,\sttwo_\nettwo)$, then:
      \begin{varitemize}
      \item
        $\sttwo_\netone$ performs an external $\delta\delta$ marriage and
        several marriages between free ports, triggering a bag of actions
        $\fmaone_\netone$ consisting of:
        \begin{varitemize}
        \item
          $(\prone_1,\maa{\delta\delta}{\skthree_1,\skfour_1})$
        \item
          $(\prone_{2i},\maa{\delta\delta}{\cnj{\skthree_{2i}},\skfour_{2i}})$ and
          $(\prone_{2i+1},\maa{\delta\delta}{\skthree_{2i+1},\skfour_{2i+1}})$
          for each $i=1,\ldots,n'$
        \item
          $(\prtwo_{2i-1},\kla{\skthree'_{2i-1},\skfour'_{2i-1}})$ and
          $(\prtwo_{2i},\klda{\skthree'_{2i},\skfour'_{2i}})$
          for each $i=1,\ldots,m'$
        \item
          $(\prthree_{2i-1},\klda{\skthree'_{2i-1},\skfour'_{2i-1}})$ and
          $(\prthree_{2i},\kla{\skthree'_{2i},\skfour'_{2i}})$
          for each $i=1,\ldots,l'$
        \end{varitemize}
      \item
        $\sttwo_\nettwo$ performs an external $\delta\delta$ marriage and
        marriages between free ports which invoke $\fmaone_\nettwo$ consisting of:
        \begin{varitemize}
        \item
          $(\prone'_{2i-1},\maa{\delta\delta}{\cnj{\skthree_{2i-1}},\skfour_{2i-1}})$ and
          $(\prone'_{2i},\maa{\delta\delta}{\skthree_{2i},\skfour_{2i}})$
          for each $i=1,\ldots,n'$
        \item
          $(\prone'_n,\maa{\delta\delta}{\skthree_n,\skfour_n})$
        \item
          $(\prtwo'_{2i-1},\klda{\skthree'_{2i-1},\skfour'_{2i-1}})$ and
          $(\prtwo'_{2i},\kla{\skthree'_{2i},\skfour'_{2i}})$
          for each $i=1,\ldots,m'$
        \item
          $(\prthree'_{2i-1},\kla{\skthree'_{2i-1},\skfour'_{2i-1}})$ and
          $(\prthree'_{2i},\klda{\skthree'_{2i},\skfour'_{2i}})$
          for each $i=1,\ldots,l'$
        \end{varitemize}
      \item
        We can see they are balanced, i.e.,
        $\fmaone_\netone\cup\fmaone_\nettwo\in\sortho{\pinjone}$.
        Thus $\sttwo$ can perform these marriages in an internal
        transition. The resulting state $\stfour$ coincides with
        $f^{\ast}(\stthree)$.
      \end{varitemize}
    \item
      If $\stthree$ is obtained from $\stone$ by performing
      an internal $\gamma\delta$ marriage, then we have two marriage tokens
      involved:
      \begin{center}
        \includegraphics[scale=1.2]{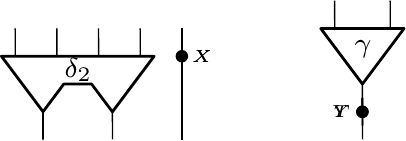}
      \end{center}
      Both of the tokens $\tkntwo$ and $\tknthree$
      originated in the $\delta_2$ cell.
      Similarly to the case of $\delta\delta$ marriage, the marriage can be
      internally simulated in $\pc{(\tkm{\netone})}{\pinjone}{(\tkm{\nettwo})}$
      via finitely many communication
      between $\tkm{\netone}$ and $\tkm{\nettwo}$:
      \begin{varitemize}
      \item The paths of $\tknone$ and $\tkntwo$ crosses the border
        finitely many times, for example as in the following figure,
        where $\pinjone(\prone_i)=\prone'_i$, $n=2n'+1$ for some $n'\geq 0$,
        and $m=2m'$ for some $m'\geq 0$.
.
        \begin{center}
          \includegraphics[scale=1.2]{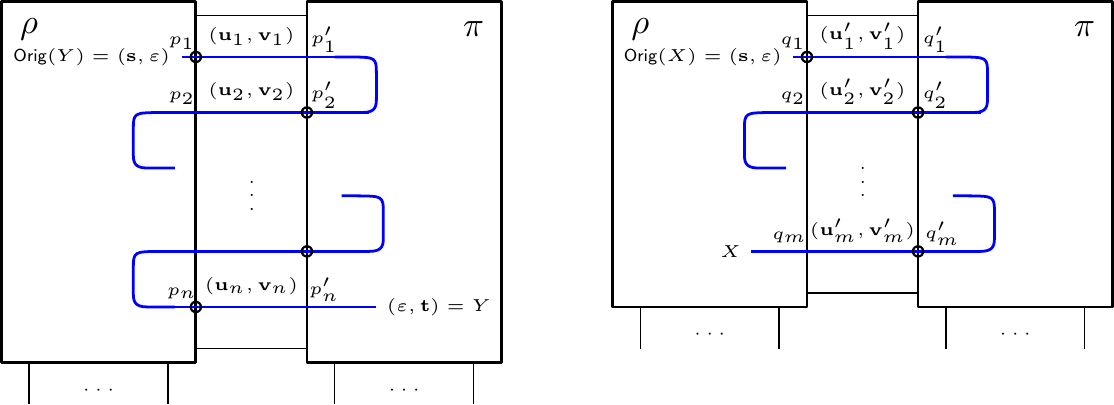}
        \end{center}
      \item A bag of actions $\sttwo_\netone$ in $\fmaone_\netone$ consisting of
        \begin{varitemize}
        \item
          $(\prone_1,\maa{\delta\gamma}{\skthree_1,\skfour_1})$
        \item
          $(\prone_{2i},\maa{\gamma\delta}{\cnj{\skthree_{2i}},\skfour_{2i}})$ and
          $(\prone_{2i+1},\maa{\delta\gamma}{\skthree_{2i+1},\skfour_{2i+1}})$
          for each $i=1,\ldots,n'$
        \item
          $(\prtwo_{2i-1},\kla{\skthree'_{2i-1},\skfour'_{2i-1}})$ and
          $(\prtwo_{2i},\klda{\skthree'_{2i},\skfour'_{2i}})$
          for each $i=1,\ldots,m'$
        \end{varitemize}
        and another  bag of actions $\sttwo_\nettwo$ in $\fmaone_\nettwo$ consisting of
        \begin{varitemize}
        \item
          $(\prone_{2i},\maa{\delta\gamma}{\skthree_{2i},\skfour_{2i}})$ and
          $(\prone_{2i-1},\maa{\gamma\delta}{\cnj{\skthree_{2i-1}},\skfour_{2i-1}})$
          for each $i=1,\ldots,n'$
        \item
          $(\prone_n,\maa{\gamma\delta}{\cnj{\skthree_1},\skfour_1})$
        \item
          $(\prtwo_{2i-1},\klda{\skthree'_{2i-1},\skfour'_{2i-1}})$ and
          $(\prtwo_{2i},\kla{\skthree'_{2i},\skfour'_{2i}})$
          for each $i=1,\ldots,m'$
        \end{varitemize}
        allow us to perform an internal marriage that yields
        a state $\stfour = \funone^*(\stthree)$.
      \end{varitemize}
    \item
      If $\stthree$ is obtained from $\stone$ by performing
      any other internal transition, it must be an internal transition
      on a marriage token that does not cross the border.
      We can easily obtain $\stfour = \funone^*(\stthree)$ by performing
      an internal transition on the corresponding marriage token
      in the same side (either in $\tkm{\netone}$ or $\tkm{\nettwo}$)
      of the token machine.
    \item
      If $\stthree$ is obtained from $\stone$ by performing
      input action $\ina{{\conone}}$, it is again easily simulated in
      $\pc{(\tkm{\netone})}{\pinjone}{(\tkm{\nettwo})}$
      by performing the same input action on the same free port.
    \item
      If $\stthree$ is obtained from $\stone$ by performing output action
      $\outa{{\conone}}$ on a token with stacks $\conone$,
      there must be a token with the same stack $\conone$ on the same free port in $\pc{(\tkm{\netone})}{\pinjone}{(\tkm{\nettwo})}$.
      We can simulate the action by performing output action on that token.
    \item The other cases (external marriages and killing/killed) can be
      simulated similarly to the internal marriage cases: a corresponding
      action and several times of communication between
      $\tkm{\netone}$ and $\tkm{\nettwo}$
      yield the state $\stfour = \funone^*(\stthree)$.
    \item The cases above extend to any external transition with multiple actions.
      This is because of the following reason.
      For each action, the involved tokens have their path
      from its origin as discussed above.
      Those paths (and matching tokens along them) are all distinct
      and thus their corresponding transitions in
      $\pc{(\tkm{\netone})}{\pinjone}{(\tkm{\nettwo})}$
      do not interfere with each other.
    \end{varitemize} 
  \item
    We then need to show that if $(\stone,\sttwo)\in\relone$ and
    $\tr{\trone}{\sttwo}{\lblone}{\stthree}$, then
    $\tr{\trtwo}{\stone}{\lblone}{\stfour}$ where
    $(\stfour,\stthree)\in\relone$.  By definition, the transition
    relation $\trone$ on
    $\pc{(\tkm{\netone})}{\pinjone}{(\tkm{\nettwo})}$ is the union of
    three relations $\trone_\netone$, $\trone_\nettwo$ and
    $\trone_{\netone\nettwo}$. Now, if
    $\tr{\trone_\netone}{\sttwo}{\lblone}{\stthree}$ or
    $\tr{\trone_\nettwo}{\sttwo}{\lblone}{\stthree}$, then finding
    $\stfour$ is trivial. If, on the other hand,
    $\tr{\trone_{\netone\nettwo}}{\sttwo}{\lblone}{\stthree}$, then
    by definition,
    \begin{align*}
      \sttwo&=(\sttwo_\netone,\sttwo_\nettwo);\\
      \stthree&=(\stthree_\netone,\stthree_\nettwo);\\ 
      \lblone\cup\lblthree&=\lbltwo_\netone\cup\lbltwo_\nettwo;\\
      &\tr{\trone}{\sttwo_\netone}{\lbltwo_\netone}{\stthree_\netone};\\
      &\tr{\trone}{\sttwo_\nettwo}{\lbltwo_\nettwo}{\stthree_\nettwo}.
    \end{align*}
    Again, if $\lblthree=\emptyset$, then finding $\stfour$ with the 
    required property is trivial because 
    $\tr{\trone}{\stone}{\lbltwo_\netone}{\stfour}$ and
    $\tr{\trone}{\stone}{\lbltwo_\nettwo}{\stfour}$, then
    $\tr{\trone}{\stone}{\lblone}{\stfour}$, since  
    $\lblone=\lbltwo_\netone\cup\lbltwo_\nettwo$. If $\lblthree$
    is \emph{not} empty, then we need to analyse several cases, keeping
    in mind that $\lblthree$ is by definition of parallel composition
    a $\pinjone$-dual multiset:
    \begin{varitemize}
    \item
      Suppose that $(\ell,\outa{\conone})\in\lblthree$. Then,
      $(\pinjone(\ell),\ina{\conone})$ is also an element
      of $\lblthree$. This can be simulated in $\stone$ by just
      an internal transition. Similarly if
      $(\ell,\ina{\conone})\in\lblthree$.
    \item
      Suppose that $(\ell,\maa{\gamma\delta}{\conone})\in\lblthree$
      as an effect of an external marriage between a $\gamma$ cell
      and a free port. Then there must also be another action
      $(\sigma(\ell),\maa{\delta\gamma}{\cnj{\conone}})$ in $\lblthree$.
      We can think of this action as being produced by a (possibly
      empty) chain of marriages between two free ports alternating
      between $\netone$ and $\nettwo$. The fact that the chain
      alternates is again due to $\pinjone$-duality of $\lblthree$.
      Let us distinguish a few further cases depending on the
      rule in which this chain ends:
      \begin{varitemize}
      \item
        If the chain ends in an \emph{external} marriage between a
        $\delta$ cell and a free port, then the whole chain, including
        its extremes, can be simulated by just an \emph{internal}
        $\gamma\delta$ marriage. If the latter comes in its version
        with a killing action in the form $(\ell',\kla{\conone})$,
        then we need to do the same kind of chain-based reasoning.
      \item
        If the chain ends in a marriage between two free
        ports one of the two not being in the domain of $\pinjone$,
        then the whole chain, including its extremes, can be simulated
        by an \emph{external} marriage between a $\gamma$ cell and
        a free port.
      \end{varitemize}
    \item
      Suppose that
      $(\ell,\maa{\gamma\delta}{\conone})\in\lblthree$
      as the effect of an external marriage between an $\delta$ and
      a free port. Then, again, we can proceed as in the last but previous
      case, with an additional complication due to the presence of
      killing actions at \emph{both} extremes of the chain.
    \item
      Suppose that
      $(\ell,\maa{\gamma\delta}{\conone})\in\lblthree$ as the
      effect of an external marriage between two free ports.
      Then we can proceed with the usual chain-based reasoning,
      but on \emph{both} sides of the action we have just identified.
    \end{varitemize}
    This way, we can ``strip-off'' all the actions from $\lblthree$
    and $\lblone$, and each of the chains of actions identified above
    corresponds to an action in $\tkm{\pc{\netone}{\nettwo}{\pinjone}}$.
  \end{varnumlist}
\end{proof}
While the Compositionality Theorem holds in single-token
machines as well~\cite{DanosRegnier}, it has not been considered in existing
work on multi-token machines~\cite{lics2014,lics2015}, partially
because the kind of nets to which the theory can be applied, the
so-called \emph{closed nets}, is quite restricted.
\section{Soundness}\label{sect:soundness}
As already mentioned several times, the transition rules for our token
machines are tailored for capturing the behaviours of interaction
rules of MICs.  This intention is unsurprisingly formalised as a
bisimulation between the machine for a net containing a redex (say
$\netone$) and the machine for the net containing its reduct (say
$\nettwo$). This, of course, must be done for each
interaction rule.  However, there is a subtlety. While the result for
$\gamma\gamma$ rule is a mere bisimulation of the two machines
$\tkm{\netone}$ and $\tkm{\nettwo}$, those for the rules involving
$\delta_2$ cells (i.e., $\gamma\delta$ and $\delta\delta$ rules) are
not; instead, the results are described as bisimulations between the
machines $\tkm{\nettwo}$ and
$\tkm{\netone}$ in which the initial state is replaced by a state
\emph{after a specific marriage transition}. At the level of the
entire machines $\tkm{\netone}$ and $\tkm{\nettwo}$, the relations are
thus simulations rather than bisimulations. This is because the token
machine for a net is designed to simulate \emph{any} possible sequence
of reductions from it, while the reduction under consideration is
\emph{a particular one} of those all possibilities, since reduction is
nondeterministic.  Therefore, a reduction on a $\delta_2$ cell always
results in discarding some otherwise possible behaviours.  This is
faithfully handled by a marriage transition in the token machine
$\tkm{\netone}$ corresponding to that reduction.  In the following
three lemmas, we study how $\tkm{\netone}$ and $\tkm{\nettwo}$ relate
whenever $\netone\red\nettwo$. This is done by giving certain
relations, that we indicate with $\relone^\netone_\nettwo$, between
the states of $\tkm{\netone}$ and those of $\tkm{\nettwo}$.  Those
relations will be useful also in Section \ref{sect:adequacy} below.

Let us first consider $\gamma\gamma$ reduction. In this case, the
machines for the redex and the reduct behave precisely the same:
\begin{lemma}\label{lemma:gammagammasoundness}
  If $\netone\pred{\gamma\gamma}\nettwo$, then
  $\tkm{\netone}\approx\tkm{\nettwo}$.
\end{lemma}
\begin{proof}
  If $\netone\pred{\gamma\gamma}\nettwo$, then we are in the following
  situation:
  \begin{center}
    \includegraphics[scale=1.2]{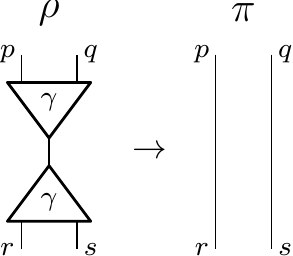}
  \end{center}
  where $p,q,r,s$ indicate the free ports of $\netone$ and $\nettwo$.
  We have a partial function $f:\tkns^\netone\rightharpoonup\tkns^\nettwo$
  such that:
  \begin{varitemize}
  \item
    Matching tokens are mapped to the same position in $\nettwo$
    if their origins are free ports.
  \item
    Marriage tokens are pushed back to their origins in $\nettwo$ if they are
    $\emptyset$-canonical. Pictorially, it can be illustrated as:
    \begin{center}
      \includegraphics[scale=1.2]{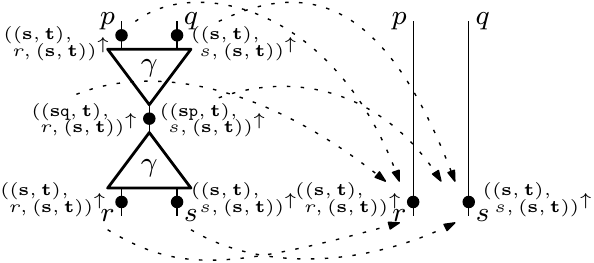}
    \end{center}
    and similarly for downward tokens.
  \end{varitemize}
  The relation $\relone$ we are looking for is defined as a pointwise extension of $f$.
  We can now show that, indeed, $\relone$ is a bisimulation relation.
  \begin{varitemize}
  \item
    Both $\isone_\netone$ and $\isone_\nettwo$ are empty. By definition, $\relone$
    relates $\isone_\netone$ and $\isone_\nettwo$.
  \item
    Assume $\stone_\netone\;\relone\;\stone_\nettwo$ and
    $\tr{\trone}{\stone_\netone}{\fmaone}{\sttwo_\netone}$. We need to find
    $\tr{\trtwo}{\stone_\nettwo}{\fmaone}{\sttwo_\nettwo}$ such that
    $\sttwo_\netone\;\relone\;\sttwo_\nettwo$. Let us proceed by analysing
    the three transition types.
    \begin{varitemize}
    \item
      If $\fmaone = \emptyset$, then the transition must be an internal move of a marriage token 
      because internal marriages never happen in $\netone$. The definition of $f$ tells us
      that $\sttwo_\netone\;\relone\;\stone_\nettwo$, so we can take $\sttwo_\nettwo=\stone_\nettwo$.
    \item
      If $\fmaone$ is an input action, let $\sttwo_\netone=\stone_\netone\uplus\mst{\tknone}$
      where $\tknone$ is the incoming token. Observe that $\stone_\nettwo$ can evolve into 
      $\sttwo_\nettwo=\stone_\nettwo\uplus\mst{\tknone}$ by performing the same input action
      and $\sttwo_\netone\;\relone\;\sttwo_\nettwo$.
    \item
      If $\fmaone$ is an output action, let
      $\sttwo_\netone=\stone_\netone\setminus\mst{\tknone}\uplus\mst{\tkntwo}$
      where $\tknone$ is the outgoing token and $\tkntwo$ is a matching token. We can see that
      $f(\tknone)$ is also ready to flow out and hence
      $\tr{\trone}{\stone_\nettwo}{\fmaone}{\stone_\nettwo\setminus\mst{f(\tknone)}\uplus\mst{\tkntwo}}$.
    \item
      If $\fmaone$ is a bunch of external marriages, then they should be marriages between
      free ports because no marriage tokens in $\stone_\netone$ are placed on $e$
      with empty $\gamma$-stacks. These actions come from matching tokens in $\stone_\netone$
      and let $\stthree$ be such tokens. They are included in $\stone_\nettwo$ as well.
      Thus we have $\tr{\trone}{\stone_\nettwo}{\fmaone}{\stone_\nettwo\setminus\stthree}$.
    \end{varitemize}
  \item
    Assume $\stone_\netone\;\relone\;\stone_\nettwo$ and
    $\tr{\trone}{\stone_\nettwo}{\fmaone}{\sttwo_\nettwo}$. As $\pi$ doesn't allow
    any internal actions, we consider external actions.
    \begin{varitemize}
    \item
      If $\fmaone$ is an input, then we can find the counterpart in $\netone$ as we did before.
    \item
      If $\fmaone$ is an output, we can suppose, without loss of generality, a marriage token
      $\tknone=(r,(\skone,\sktwo),r,(\skone,\sktwo))\in\stone_\nettwo$ flows out through $p$.
      Then $\stone_\netone$ has a marriage token which belongs to $f^{-1}(\tknone)$.
      The token can reach $p$ with a configuration $(\skone,\sktwo)$ after a few internal
      transitions and its origin is the same as $\tknone$. So it allows $\tkm{\netone}$ to
      perform the compatible output action.
    \item
      If $\fmaone$ consists of marriages, $\tkm{\netone}$ is able to consume 
      the corresponding matching tokens in $\stone_\netone$.
    \end{varitemize}
  \end{varitemize}
\end{proof}
If one of the cells involved in a reduction is a $\delta_2$ cell, one cannot hope
to get a result as strong as Lemma~\ref{lemma:gammagammasoundness}:
\begin{lemma}\label{lemma:gammadeltasoundness}
  If $\netone\pred{\gamma\delta}\nettwo$, then
  there is $\istwo_\netone$ such that $\isone_\netone\red\istwo_\netone$
  by an internal marriage transition and
  $(\fmsts{\tkns^\netone},\red,\istwo_\netone)\approx\tkm{\nettwo}$.
\end{lemma}
\begin{proof}
  We are in the following situation:
  \begin{center}
    \includegraphics[scale=1.2]{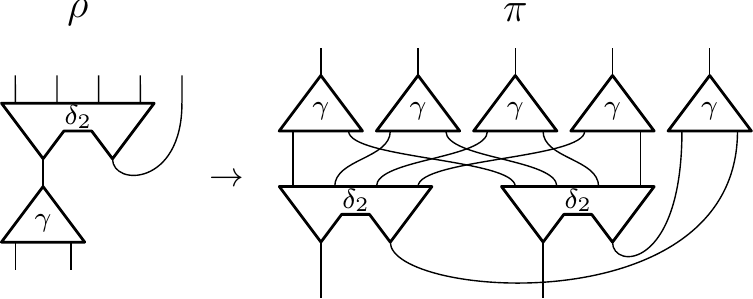}
  \end{center}
  Observe that $\isone_\netone$ can perform a marriage between the $\gamma$ and $\delta_2$ cells,
  resulting in a state $\istwo_\netone$ which has 1 married status tokens, 5 single status tokens
  and 4 marriage tokens. We denote the 2 status tokens in $\istwo_\netone$ which have empty $\gamma$
  stacks by $\stthree$. We define a relation $\relone$ between $\fmsts{\tkns^\netone}$ and
  $\fmsts{\tkns^\nettwo}$, which is described in a pointwise manner as follows:
  \begin{varitemize}
  \item
    Status tokens in $\netone$ whose topmost symbols are $\lft$ (resp. $\rgt$) are related to ones 
    of the $\delta$ cell on the left (resp. on the right) in $\nettwo$:
    \begin{center}
      \includegraphics[scale=1.2]{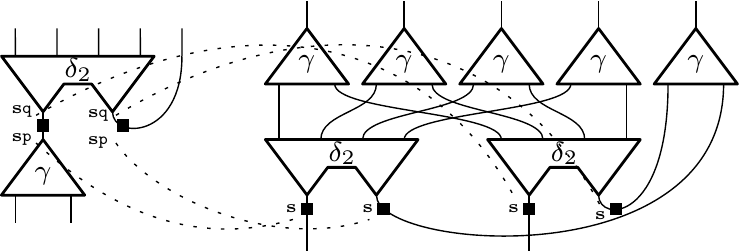}
    \end{center}
    and similarly for married status tokens. 
  \item
    Marriage tokens originating from principal ports are related if their origins are
    related in a similar way to the one described above. Marriage tokens originating
    from free ports are related if they have the same origins.
  \item
    Matching tokens are related similarly to marriage tokens. 
  \end{varitemize}
  Then let $\relone'=\{(\stone_\netone,\stone_\nettwo)
  \mid(\stone_\netone-\stthree,\stone_\nettwo)\in\relone,
  \stone_\netone\text{ and }\stone_\nettwo\text{ are canonical}\}$.
  We claim $\relone'$ is a bisimulation.
  \begin{varitemize}
  \item
    We can see that $\relone'$ relates $\istwo_\netone$ and $\isone_\nettwo$.
  \item
    Suppose $(\stone_\netone,\stone_\nettwo)\in\relone'$
    and $\tr{\trone}{\stone_\netone}{\fmaone}{\sttwo_\netone}$.
    \begin{varitemize}
    \item
      If $\fmaone=\emptyset$, then one of marriage tokens in $\stone_\netone$ moved. But
      the marriage token and its counterpart in $\stone_\nettwo$ are still related and
      so $\stone_\nettwo$ can remain unchanged. 
    \item
      If $\fmaone$ is an input action, it is easy to see that making the same input in $\nettwo$
      will lead $\stone_\nettwo$ to a state related to $\sttwo_\netone$.
    \item
      If $\fmaone$ is an output action, let $\tknone$ be the outgoing token in $\stone_\netone$
      and $\tknone'$ be the token in $\stone_\nettwo$ related to $\tknone$.
      There are two cases.
      \begin{varitemize}
      \item
        If the origin of $\tknone$ is a principal port, then the pair of $\tknone$ and $\tknone'$ 
        is one of the following pairs:
        \begin{center}
          \includegraphics[scale=1.2]{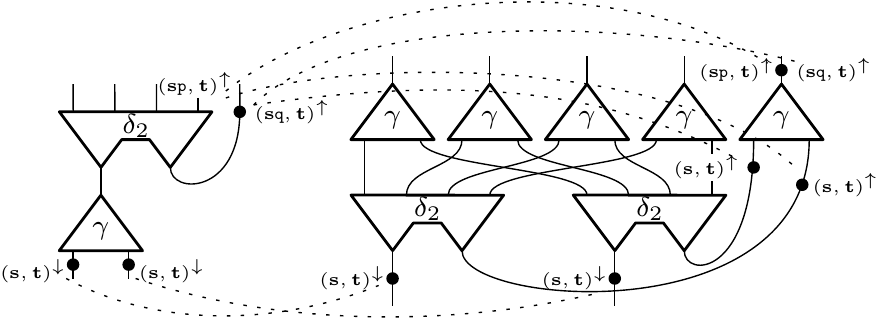}
        \end{center}
        In any case, $\tknone'$ can reach the same position and configuration as $\tknone$
        within zero or one transition steps. Thus we have
        $\tr{\trtwo}{\stone_\nettwo}{\fmaone}{\sttwo_\nettwo}$.
      \item
        If the origin of $\tknone$ is a free port, then $\tknone'$ is on this free port. By
        definition of $\relone'$, $\stone_\netone$ has a status tokens $\tknthree$ which is
        necessary for $\tknone$ to make a trip from its origin to its current position.
        The counterpart $\tknthree'$ for $\tknthree$ is contained in $\stone_\nettwo$ and
        indeed helps $\tknone'$ to go to the correct location. For instance, let us consider
        the situation below, in which we have
        $\tknone=(r, (\skone, \sktwo\mathtt{1}), p, (\skone\lft, \sktwo))$,
        $\tknone'=(r, (\skone\lft, \sktwo)), p, (\skone\lft, \sktwo))$ and
        $\tknthree=(q, (\skone\lft, \gamma))$. Notice that
        $\tknone'$ can arrive at $r$ thanks to $\tknthree'=(q_l, (\skone, \gamma))$.
        \begin{center}
          \includegraphics[scale=1.2]{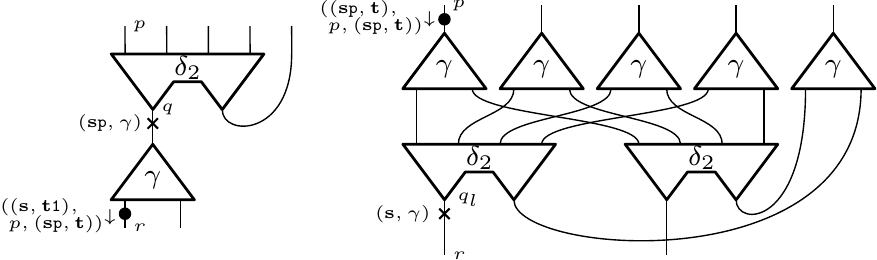}
        \end{center}
      \end{varitemize}
    \item
      If $\fmaone$ is a multiset of external marriages, each element of $\fmaone$ is either
      a marriage between a principal port and a free port or one between two free ports.
      \begin{varitemize}
      \item
        A marriage between two free ports involves one or two matching tokens whose origins
        are free ports. Since such matching tokens in $\stone_\netone$ are related to their
        copies in $\stone_\nettwo$, $\stone_\nettwo$ can imitate that marriage.
      \item
        A marriage between a principal port and a free port provokes one marriage action,
        whose marriage type is either $\gamma\delta$, $\delta\delta$ or $\delta\gamma$.
        The pictures below illustrate how $\nettwo$ simulates the external marriage in
        $\netone$ for each marriage type.
        \begin{varitemize}
        \item
          $\gamma\delta$
          \begin{center}
            \includegraphics[scale=1.2]{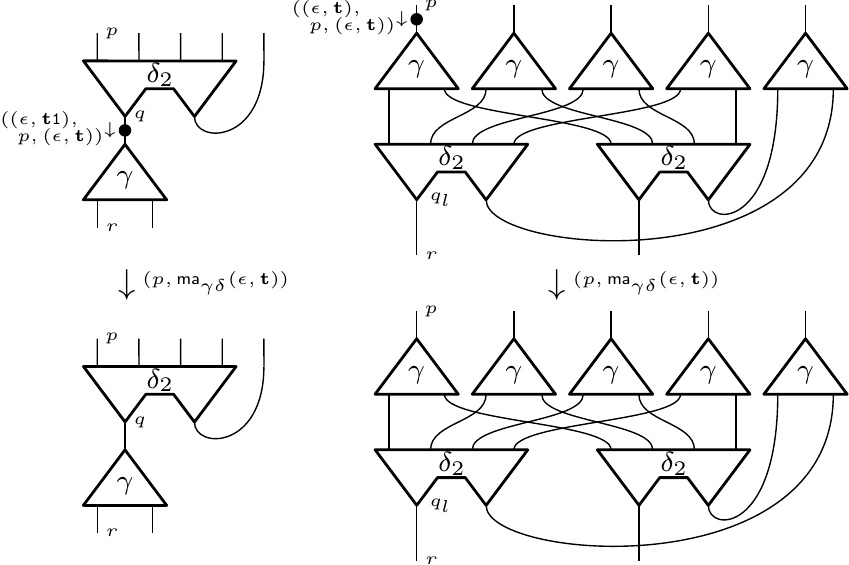}
          \end{center}
        \item
          $\delta\delta$
          \begin{center}
            \includegraphics[scale=1.2]{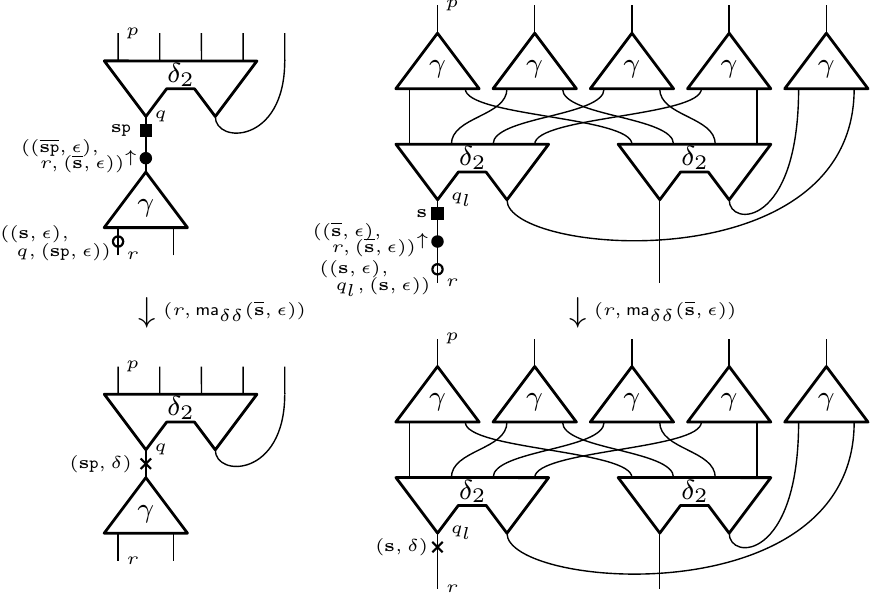}
          \end{center}
        \item
          $\delta\gamma$ 
          \begin{center}
            \includegraphics[scale=1.2]{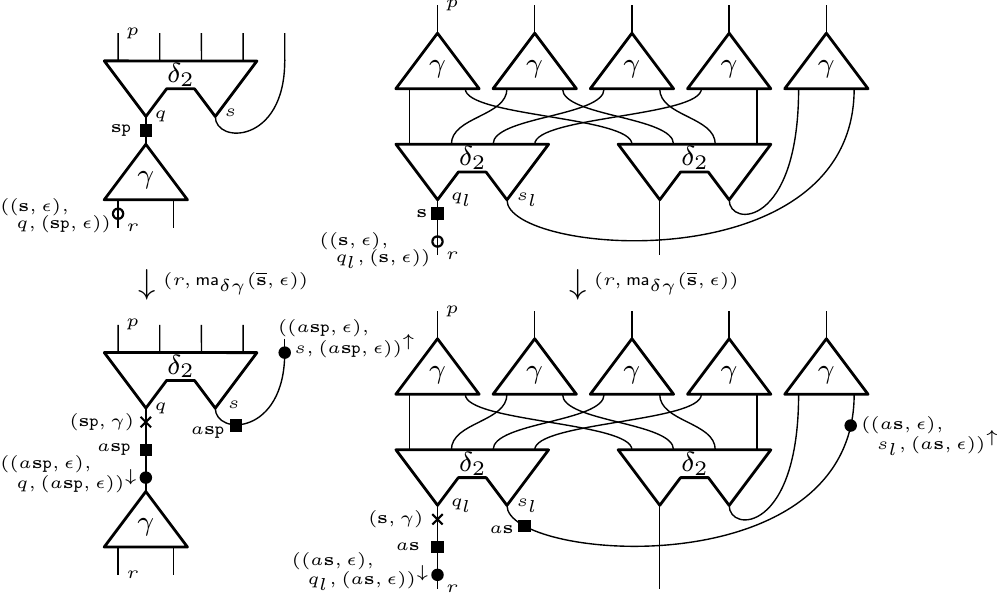}
          \end{center}
        \end{varitemize}
      \end{varitemize}
    \end{varitemize}
  \item
    For the other direction, we can proceed similarly.
  \end{varitemize}
\end{proof}
\begin{lemma}\label{lemma:deltadeltasoundness}
  If $\netone\pred{\delta\delta}\nettwo$, then
  there is $\istwo_\netone$ such that $\isone_\netone\red\istwo_\netone$
  by an internal marriage transition and
  $(\fmsts{\tkns^\netone},\red,\istwo_\netone)\approx\tkm{\nettwo}$.
\end{lemma}
\begin{proof}
  Assume that the reduction is an interaction between two first
  principal ports of two $\delta_2$ cells.
  \begin{center}
    \includegraphics[scale=1.2]{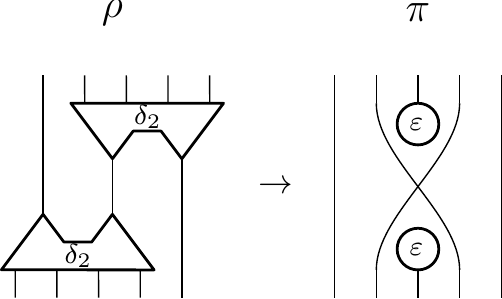}
  \end{center}
  We can see $\isone_\netone$ immediately evolves into
  $\istwo_\netone$, which consists of 2 married status tokens and 2
  single status tokens. We define a relation
  $\relone\subseteq\fmsts{\tkns^\netone}\times\fmsts{\tkns^\nettwo}$
  as:
  \begin{varitemize}
  \item
    Matching tokens are related naturally.
  \item
    Marriage tokens in $\netone$ which originated in free ports and
    are $\istwo_\netone$-canonical are related to the natural
    counterparts in $\nettwo$ for their origins.
  \end{varitemize}
  The bisimulation between $\istwo_\netone$ and $\isone_\nettwo$ is
  given by
  $\relone'=\{(\stone_\netone\uplus\istwo_\netone,\stone_\nettwo)
  \mid(\stone_\netone,\stone_\nettwo)\in\relone\}$. The proof that
  $\relone'$ is a bisimulation is very much similar to the one of
  Lemma~\ref{lemma:gammagammasoundness}.
\end{proof}

One may wonder why a $\gamma\delta$ reduction is treated the same as a
$\delta\delta$ reduction, since no nondeterminism seems to be involved
in it. This is however, only apparent: if a $\delta_2$ cell may
``choose'' to interact with either a $\delta_2$ or a $\gamma$ cell, it
does not loose any behaviour \emph{only up to $\eta$-equivalence}
(see~\cite{MazzaThesis} for a definition) when choosing to interact
with the $\gamma$ cell. GoI, on the other hand, is well-known
\emph{not} to be sound for $\eta$ in general. In other words, although
$\gamma\delta$ reduction is \emph{deterministic} as far as rewriting
on nets is concerned, it is not so at the level of the underlying GoI
model, and this happens for deep reasons which have little to do with
MICs.

As a corollary of the three lemmas above, we can prove soundness
of our model with respect to net reduction, which is spelled out
as a similarity:
\begin{corollary}[Soundness]\label{cor:soundness}
  If $\netone\red\nettwo$,
  then $\tkm{\nettwo}\precapprox\tkm{\netone}$.
\end{corollary}
\begin{proof}
  This 
  follows from
  Lemma~\ref{lemma:gammagammasoundness},
  Lemma~\ref{lemma:gammadeltasoundness}, and
  Lemma~\ref{lemma:deltadeltasoundness}. Indeed,
  if $\isone_\netone\red\istwo_\netone$, then
  $(\fmsts{\tkns^\netone},\red,\istwo_\netone)\precapprox\tkm{\netone}$.
\end{proof}
\section{Adequacy}\label{sect:adequacy}
\newcommand{\tsone}{\mathbf{A}} \newcommand{\tstwo}{\mathbf{B}}
\newcommand{\tsthree}{\mathbf{C}}
\newcommand{\fincomp}[3]{\mathit{fin}_{#1}^{#2}(#3)}
\newcommand{\infcomp}[3]{\mathit{inf}_{#1}^{#2}(#3)}
\newcommand{\setnfone}{\mathbf{X}} \newcommand{\netnfvic}{\mathbf{NV}}
\newcommand{\netvic}{\mathbf{VC}}
\newcommand{\tkmnf}[1]{\mathbf{TN}_{#1}}
Corollary~\ref{cor:soundness} tells us that, along any interaction
path (of which there could be many) starting from any net $\netone$,
one finds nets whose token machines can all be related by way of weak
similarity to $\tkm{\netone}$.  This already tells us much about the
way nets and token machines are related.  However, there is still no
result around about the relationship between the behaviour of
$\tkm{\netone}$ and that of $\netone$ itself, e.g., any result on
whether there is a way to ``read'' a property of $\netone$ from its
interpretation $\tkm{\netone}$. This section is devoted to proving
that, indeed, $\netone$ and its token machine have the same behaviour
as for \emph{termination}. This means that token machines are not only
a sound but also an \emph{adequate} model of multiport interaction,
since termination is the most natural property to be observed.

Before stating adequacy, we need to introduce some more preliminary
definitions and notations. In particular, in this section, contrarily
to the previous ones, token machines will be seen as non-interactive
objects, and thus analysed as \emph{closed} systems.  A
\emph{(unlabelled) transition system} (TS for short; also known as a
\emph{Kripke structure}) consists of a set $\ssone$ of \emph{states}
and a binary relation ${\red}\subseteq\ssone\times\ssone$.  A
\emph{pointed} transition system is a triple $(\ssone,\red,\isone)$
where $(\ssone,\red)$ is a TS and $\isone\in\ssone$ is the
\emph{initial state}. Notions of simulation and bisimulation can
easily be given for pointed transition systems. Noticeably, in a
simulation we require each transition in the simulated TS to
correspond to \emph{zero or more} transitions in the simulating TS.
We indicate the obtained notions of (weak) bisimilarity and similarity
as $\approx$ and $\precapprox$, as usual. Given a transition system
$\tsone=(\ssone,\red)$, we define the following two predicates
parametrised by an element $\stone$ of $\ssone$ and by a set of
elements $\setnfone$:
\begin{varitemize}
\item
  $\fincomp{\tsone}{\setnfone}{\stone}$, which holds iff \emph{there is}
  a \emph{finite} transition sequence starting in $\stone$ and
  ending in a normal form \emph{not in} $\setnfone$.
\item
  $\infcomp{\tsone}{\setnfone}{\stone}$, which instead holds iff \emph{there is}
  an \emph{infinite} transition sequence starting in $\stone$, or 
  a finite transition sequence starting in $\stone$ and ending
  in an element \emph{in} $\setnfone$.
\end{varitemize}
$\netset$ can be seen as a transition system whose states are nets, and
whose transition relation is the one induced by the reduction
relation. Let $\netvic$ be the set of nets containing vicious circles.
For any net $\netone$, the set $\tkm{\netone}$ can itself be seen as a transition
system whose states are elements of $\tkns^\netone$ and whose
transition relation consists of internal transitions.
\begin{theorem}[Adequacy]\label{thm:adequacy}
For every net $\netone$, the following holds:
\begin{align*}
  \fincomp{\tkm{\netone}}{\emptyset}{\isone_{\netone}}&\Longleftrightarrow\fincomp{\netset}{\netvic}{\netone};\\
  \infcomp{\tkm{\netone}}{\emptyset}{\isone_{\netone}}&\Longleftrightarrow\infcomp{\netset}{\netvic}{\netone}.
\end{align*}
\end{theorem}
 The rest of this section will be devoted to proving the four
 implications which together form Theorem~\ref{thm:adequacy}.
\subsection{Standard Computations}
\newcommand{\thunkof}[2]{\mathcal{THU}^{#1}_{#2}}
In this section we define a key notion, called \emph{standard computation},
which plays a crucial role to prove Adequacy.
In the following, a \emph{computation} means a transition sequence
starting in the initial state of the underlying token machine.

  A \emph{thunk} is any transition sequence in one of two forms:
  \begin{varitemize}
  \item
    Either in the form $\fcone,\fctwo,\fcthree$, where
    \begin{align*}
      \fcone&:\tknone_1\fred{\stone}\tknone_2\fred{\stone}\cdots\fred{\stone}\tknone_n;\\
      \fctwo&:\tkntwo_1\fred{\sttwo}\tkntwo_2\fred{\sttwo}\cdots\fred{\sttwo}\tkntwo_m
    \end{align*}
    where $\tknone_1$ and $\tkntwo_1$ are origin tokens,
    and 
    $\fcthree$ is a $\delta\delta$ marriage 
    whose main tokens are $\tknone_n$ and $\tkntwo_m$;
  \item
    Or in the form $\fcone,\fctwo$, where
    $$
    \fcone:\tknone_1\fred{\stone}\tknone_2\fred{\stone}\cdots\fred{\stone}\tknone_n
    $$
    $\tknone_1$ is an origin token, and $\fctwo$ is a $\gamma\delta$ marriage
    whose main token is $\tknone_n$.
  \end{varitemize}
  A \emph{standard sequence} is a transition sequence in one of two forms:
  \begin{varitemize}
  \item
    Either in the form $\fcone_1,\ldots,\fcone_n,\fctwo$ where
    $\fcone_1,\ldots,\fcone_n$ are all thunks, and $\fctwo$ is any \emph{maximal}
    transition sequence \emph{not} containing any marriage step,
    called a \emph{continuation}.
  \item
    Or in the form $\fcone_1,\fcone_2,\ldots$, where all the $\fcone_i$
    (where $i\in\NN^+$) is a thunk.
  \end{varitemize}
A \emph{standard computation} is a computation which is a standard sequence.
The key insight is that any maximal computation can be ``represented'' by
a standard computation. Formally, we have the following lemma:
\begin{lemma}[Standardisation]
\begin{varitemize}
\item Let $\fcone\colon\isone_{\tkm{\netone}}\red\stone_1\red\stone_2\red\cdots\red\stone_n$ be a finite maximal computation of $\tkm{\netone}$.
  There exists a finite standard computation
  $\fctwo\colon\isone_{\tkm{\netone}}\red\sttwo_1\red\sttwo_2\red\cdots\red\sttwo_m$
  s.t.\ $\sttwo_m = \stone_n$.
\item Let $\fcone:\isone_{\tkm{\netone}}\red\stone_1\red\stone_2\red\cdots$ be an infinite maximal computation of $\tkm{\netone}$.
  There exists an infinite standard computation
  $\fctwo\colon\isone_{\tkm{\netone}}\red\sttwo_1\red\sttwo_2\red\cdots$
  s.t.\ for any state $\stone_i$ right after a marriage in $\fcone$,
  there exists a state $\sttwo_{i'} = \stone_i$ in $\fctwo$.
\end{varitemize}
\end{lemma}
\begin{proof}
First we show the following auxiliary propositions.
\begin{proposition}\label{prop:convert}
  In any token machine, the following three hold.
  \begin{varnumlist}
  \item \label{item:convert1}
    Let $\stone_0\red\stone_1$ be a move on a token $\tknone$
    and $\stone_1\red\stone_2$ be a move on another $\tkntwo \neq \tknone$.
    Then there exists a state $\stone'_1$ satisfying
    $\stone_0\red\stone'_1\red\stone_2$, where
    $\stone_0\red\stone'_1$ is a move on $\tkntwo$
    and $\stone'_1\red\stone_2$ is a move on $\tknone$.
  \item \label{item:convert2}
    Let $\stone_1\red\stone_2$ be a marriage
    and $\stone_0\red\stone_1$ be a move on a token $\tknone$
    not involved in the marriage $\stone_1\red\stone_2$.
    Then there exists a state $\stone'_1$ satisfying
    $\stone_0\red\stone'_1\red\stone_2$,
    where $\stone_0\red\stone'_1$ is a marriage on the same token(s) as
    the one $\stone_1\red\stone_2$,
    and $\stone'_1\red\stone_2$ is a move on $\tknone$.
  \item \label{item:convert3}
    Let $\stone_1\red\stone_2$ be a marriage
    and $\stone_0\red\stone_1$ be a move on a token $\tknone$
    killed by the marriage $\stone_1\red\stone_2$.
    Then $\stone_0\red\stone_2$, by a marriage on the same token(s).
  \end{varnumlist}
\end{proposition}
\begin{proof}
  \begin{varnumlist}
    \item
      Any move transition does not affect the other token's move transition.
      This is in particular because any move transition does not put any
      marriage status token by definition.
    \item
      Since $\stone_0\red\stone_1$ is already allowed to happen before the marriage $\stone_1\red\stone_2$,
      there is already some marriage status token(s) that allow $\tknone$ to move
      in $\stone_0$.
      Since any marriage does not delete a marriage status token by definition,
      the move is still possible in $\stone'_1$.
    \item
      Since $\stone_0\red\stone_1$ is a move, it does not put any new token nor
      change any other token's stacks by definition.
  \end{varnumlist}
\end{proof}
Let $\fcone\colon\stone_0\red\stone_1\red
\cdots\red\stone_k\red\stone_{k+1}\red\stone_{k+2}\red\cdots\red\stone_n$ 
be a finite maximal transition sequence where $\stone_k\red\stone_{k+1}$
is the first marriage in the sequence $\fcone$.
We first permute $\fcone$, by repeatedly applying
Proposition~\ref{prop:convert}.\ref{item:convert1}, into another transition sequence
$\fcone'\colon\stone_0\red^*\stone'_1\red^*\stone'_2\red^*\stone_k
\red\stone_{k+1}\red\stone_{k+2}\red\cdots\red\stone_n$,
where $\stone_0\red^*\stone'_1$ is one (if the marriage is $\gamma\delta$)
or two (if the marriage is $\delta\delta$) consecutive focused sequences
on the token(s) involved in the marriage,
$\stone'_1\red^*\stone'_2$ consists only of moves on tokens not involved in
the marriage,
$\stone'_2\red^*\stone_k$ consists only of moves on tokens killed by the marriage,
and the rest is the same as $\fcone$.
Then, applying Proposition~\ref{prop:convert}.\ref{item:convert2} and
\ref{prop:convert}.\ref{item:convert3}, we obtain a transition sequence
$\fcone''\colon\stone_0\red^*\stone'_1\red\stone'_2
\red^*\stone_{k+1}\red\stone_{k+2}\red\cdots\red\stone_n$,
where $\stone_0\red^*\stone'_1$ is focused sequences on
the token(s) involved in the marriage $\stone'_1\red\stone'_2$,
i.e.\ with a thunk at its head.
If the subsequence
$\stone'_2\red^*\stone_{k+1}\red\stone_{k+2}\red\cdots\red\stone_n$
still contains marriages, apply the procedure described above inductively.

\end{proof}
\subsection{Preservation Properties}
\newcommand{\crone}{\mathit{RED}}
  Given $\netone\red\nettwo$, one can define a relation
  $\relone^{\netone}_{\nettwo}$ and prove it having certain
  bisimulation properties, as described in Section~\ref{sect:soundness}.
  The relations are actually useful here, too; in this section,
  we show various properties on a reduction and transition sequences using 
  the relations $\relone^{\netone}_{\nettwo}$.
  By definition, the relations $\relone^{\netone}_{\nettwo}$ are all defined
  as set of pairs $(\stone_\netone,\stone_\nettwo)$, where
  marriage tokens in $\stone_\netone$ have a natural counterpart
  in $\stone_\nettwo$, and \emph{vice versa}.
  The fact that $\relone^{\netone}_{\nettwo}$ are weak bisimulation
  relations implies that if
  $(\stone_\netone,\stone_\nettwo)\in\relone^{\netone}_\nettwo$ and
  $\fcone$ is a transition sequence starting in $\stone_\netone$ and
  ending in $\sttwo_\netone$, then there are some naturally defined
  sequences starting in $\stone_\nettwo$ and ending in
  $\sttwo_\nettwo$ such that
  $(\sttwo_\netone,\sttwo_\nettwo)\in\relone^{\netone}_\nettwo$. In
  this case we say that $\relone^{\netone}_{\nettwo}$ \emph{sends}
  $\fcone$ to any of those transition sequences or, simply, that
  $\fcone$ \emph{is sent} to any of them.  Conversely if $\fctwo$ is a
  transition sequence starting in $\stone_\nettwo$ and ending in
  $\sttwo_\nettwo$, then there are some naturally defined sequences
  starting from $\stone_\netone$ and ending in $\sttwo_\netone$ such
  that
  $(\sttwo_\netone,\sttwo_\nettwo)\in\relone^{\netone}_\nettwo$. In
  this case we say that $\relone^{\netone}_{\nettwo}$ \emph{sends
    back} any of these sequences to $\fcone$ or, simply, that any of
  these sequences \emph{is sent back} to $\fcone$. All this can
  be easily generalised to infinite transition sequences.
  If a transition sequence $\fcone$ in a collection $\crone$ is sent
  (sent back, respectively) to a transition $\fctwo$, not much can be
  said, in general, about whether $\fctwo\in\crone$. If, however,
  \emph{for every} $\fcone\in\crone$, \emph{there is}
  $\fctwo\in\crone$ such that $\fcone$ is sent (sent back,
  respectively) to $\fctwo$, we say that $\crone$ (or the predicate on
  transition sequences defining $\crone$) is \emph{preserved forward}
  (\emph{preserved backward}, respectively).
\begin{lemma}\label{lemma:inffocred}
  Infinite focused sequences starting in
  reachable states are preserved forward and backward.
\end{lemma}
\begin{proof}
  Suppose that $\netone\red\nettwo$. Then, $\netone$ can be
  seen as a redex $\netthree$ in parallel to the rest of the
  net $\netfour$, while $\nettwo$ can be seen as a reduct $\netfive$
  in parallel to $\netfour$, graphically:
  \begin{flushleft}
    \includegraphics[scale=1.5]{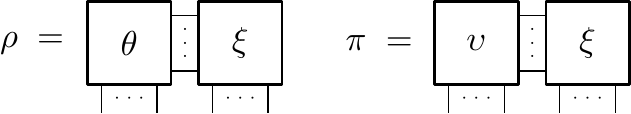}
  \end{flushleft}
  The redex $\netthree$ can be any pair of cells facing each other,
  together with wires linking some of the auxiliary ports of the two
  cells, if this is the way they are linked in $\netone$.
  Now, consider an infinite focused sequence $\fcone$ in
  $\netone$. We can have two cases:
  \begin{varitemize}
  \item
    $\fcone$ goes back and forth between $\netthree$ and $\netfour$
    \emph{infinitely} often. In this case, it is easy to realise that
    $\fcone$ is sent to an infinite focused sequence in
    $\nettwo$. Indeed, each of the infinitely many finite portions of
    it in $\netthree$ (respectively, $\netfour$) is mapped to a finite
    corresponding portion in $\netfive$ (respectively, $\netfour$).
    Moreover, each of the finite portions of $\fcone$ in $\netfour$ is
    non-empty (because if it is empty, the corresponding wire would be part
    of $\netthree$) and is sent to itself.
  \item
    $\fcone$ goes back and forth between $\netthree$ and $\netfour$
    only \emph{finitely} often. In this case, there needs to be an
    infinite suffix $\fctwo$ of $\fcone$ which lies in one of the two
    components, which however must be $\netfour$, because no infinite
    focused sequence starting in a reachable state can be constructed
    for a redex. As a consequence, $\fcone$ is sent forward to an
    infinite focused sequence in $\nettwo$, because $\fctwo$ is also
    a transition sequence of $\nettwo$.
  \end{varitemize}
  A similar argument can be given for every infinite focused
  sequence $\fcone$ in $\nettwo$.
\end{proof}
\begin{lemma}
  Thunks are preserved forward and backward.
\end{lemma}
\begin{proof}
  The first thing to observe here is that the left projection
  of $\relone^{\netone}_\nettwo$ only includes states in which
  the marriage corresponding to the redex leading $\netone$
  to $\nettwo$ has already been performed. Once we realise
  this, showing that thunks are preserved forward and backward
  amounts to the following observations:
  \begin{varitemize}
  \item
    First of all, $\relone^{\netone}_\nettwo$ puts origin tokens
    in correspondence with origin tokens, and states for which
    marriages are ready to states with the same property.
  \item
    Second of all, focused sequences are sent (and sent back)
    to focused sequences.
  \end{varitemize}
  This concludes the proof.
\end{proof}
\begin{lemma}
  Continuations starting in a reachable state are preserved backward
  and forward.
\end{lemma}
\begin{lemma}\label{lem:infContPreservation}
  Infinite continuations starting in a reachable state are preserved
  backward and forward.
\end{lemma}
\begin{proof}
  We can first of all observe that continuations are preserved
  backward and forward. Now, suppose by way of contradiction, that an
  \emph{infinite} continuation $\fcone$ is sent to a \emph{finite}
  continuation $\fctwo$. Since in any continuation, marriages are not
  allowed to happen, it means that $\fcone$ (respectively, $\fctwo$)
  can be seen as the interleaving of a finite number of focused
  reductions $\fcone_1,\ldots,\fcone_n$ (respectively,
  $\fctwo_1,\ldots,\fctwo_n$) such that $\fcone_i$ is mapped to
  $\fctwo_i$. Now, at least one of the $\fcone_i$, say $\fcone_j$ must
  be infinite, otherwise $\fcone$ would be finite itself.  By
  Lemma~\ref{lemma:inffocred}, $\fctwo_j$ is thus infinite, and as a
  consequence the whole of $\fctwo$ is infinite.
\end{proof}
\begin{proposition}\label{prop:prescompone}
  Suppose that $\netone\pred{a}\nettwo$ where
  $a\in\{\delta\delta,\gamma\delta\}$. If $\tkm{\nettwo}$ has a
  finite standard computation with $n$ thunks then $\tkm{\netone}$ has
  a finite standard computation with $n+1$ thunks.
\end{proposition}
\begin{proof}
  By Lemma~\ref{lemma:gammadeltasoundness} and
  Lemma~\ref{lemma:deltadeltasoundness}, there is $\istwo_\netone$
  such that $\isone_\netone\red\istwo_\netone$ and
  $(\istwo_\netone,\isone_\nettwo)\in\relone^\netone_\nettwo$. Since
  thunks and continuations are preserved backward, the thesis easily
  follows.
\end{proof}
\begin{proposition}\label{prop:prescomptwo}
  Suppose that $\netone\pred{\gamma\gamma}\nettwo$.
  If $\tkm{\nettwo}$ has a standard computation with $n$ thunks, then
  $\tkm{\netone}$ also has one.
\end{proposition}
\begin{proof}
  By Lemma~\ref{lemma:gammagammasoundness}, it holds
  that $\tkm{\netone}\approx\tkm{\nettwo}$. Since thunks and continuations
  are preserved backward, the thesis easily follows.
\end{proof}
\begin{proposition}\label{prop:finstdzerofinite}
  If $\tkm{\netone}$ has a finite standard computation which is a
  continuation, then $\netone\red^*\nettwo$, where 
  $\nettwo$ is in normal form and $\nettwo\notin\netvic$.
\end{proposition}
\begin{proof}
  Consider the net $\nettwo$ obtained by reducing all $\gamma\gamma$
  redexes in $\netone$, in an arbitrary order. This process of
  course terminates, because each reduction step makes the number of
  wires in the net to decrease. Now:
  \begin{varitemize}
  \item
    $\nettwo$ is in normal form. If it contains a $\delta\delta$ or
    $\delta\gamma$ redex, indeed, $\tkm{\nettwo}$ would admit a
    standard computation with at least one thunk. Since thunks are
    preserved backward, one would get the same for $\tkm{\netone}$.
    This is however incompatible with $\tkm{\netone}$ having a
    finite standard computation which is a continuation.
  \item
    $\nettwo$ does not contain any $\gamma$-vicious circle.
    If it contains one, indeed, $\tkm{\nettwo}$ would admit an \emph{infinite}
    standard computation which is a continuation, and then also
    $\tkm{\netone}$ would admit one. This is again incompatible with
    $\tkm{\netone}$ having a finite standard computation which is
    a continuation.
  \end{varitemize}
  This concludes the proof.
\end{proof}
\begin{proposition}\label{prop:finstdsucc}
  If $\tkm{\netone}$ has a standard computation with $n+1$ thunks and
  a continuation $\fcone$ (respectively, with infinitely many thunks),
  then $\netone\red^+\nettwo$, where $\tkm{\nettwo}$ has a standard
  computation with $n$ thunks and a continuation
  $\fctwo$ (respectively, with infinitely many thunks). Moreover
  $\fcone$ is infinite iff $\fctwo$ is.
\end{proposition}
\begin{proof}
  First of all, let us consider the net $\netthree$ one obtains by
  reducing all $\gamma\gamma$ redexes in $\netone$, in an arbitrary
  order. This process of course terminates, because each reduction
  step makes the number of wires in the net to decrease. The net
  $\netthree$ has itself a standard computation with $n+1$ thunks
  and a continuation (respectively, it has infinitely many thunks).
  Since $\netthree$ only contains $\delta\delta$ or $\delta\gamma$
  redexes, the first thunk in the standard computation corresponds
  to one of redex of one of those kinds. As a consequence
  there is $\nettwo$ such that $\netthree\red\nettwo$. Moreover,
  by Lemma~\ref{lemma:gammadeltasoundness} and
  Lemma~\ref{lemma:deltadeltasoundness}, $\nettwo$ admits
  a standard computation $\fctwo$ with $n$ thunks and a continuation
  (respectively, with infinitely many thunks), which is
  infinite iff $\fcone$ is.
\end{proof}
\begin{proposition}\label{prop:netnormalform}
  Let $\netone$ be a net in normal form.
  If 
  $\netone\notin\netvic$,
  then $\tkm{\netone}$ only admits finite computations. Otherwise, all
  maximal computation of $\tkm{\netone}$ are infinite.
\end{proposition}
\begin{proof}
Assume $\netone\notin\netvic$.
In this situation, a marriage token cannot hit any principal port,
since it means there remains a redex (and thus there will be no marriage).
Thus a marriage token can only hit auxiliary ports one by one;
if it hits an auxiliary port of a $\delta_2$ cell, it stops.
Moreover there is no way to make it able to move on, since no marriage can take place.
Therefore the only possibility for a token to travel infinitely is to go through $\gamma$ cells infinitely many times, each time from an auxiliary port to the principal port.
However, since $\netone$ is finite, this means the token is trapped into a loop consisting of $\gamma$ cells, which means there is a $\gamma$-vicious circle, which in turn contradicts to $\netone\notin\netvic$.
Hence a marriage token can perform transitions only finitely many times.
Considering that the number of marriage tokens is also finite and it does not increase, all computations of $\tkm{\netone}$ are finite.

Assume $\netone\in\netvic$. Recall it means $\netone$ contains a subnet in the following form:

    \begin{center}
      \includegraphics[scale=1.2]{viciouscircle}
    \end{center}

Since $\nu$ contains one or more $\delta_2$ cells,
they emit marriage tokens and
at least one of them reaches $\sigma$ since $\nu$ is a tree.
It reaches one of the leaves of $\tau$ since $\sigma$ is just a wiring.
Finally, the token necessarily reaches the root of $\tau$ and is feeded back to $\sigma$, since $\tau$ consists only of $\gamma$ cells; this continues forever.
By the same reason as the case of $\netone\notin\netvic$, there will not be any marriage, hence the infinite trip described above cannot end by being killed.
\end{proof}

\begin{proposition}\label{prop:infCompImpliesVC}
If $\tkm{\netone}$ has an infinite standard computation which is a
continuation, then there exists $\nettwo$ satisfying
$\netone \red^* \nettwo$ solely by $\gamma\gamma$ reductions and
$\nettwo \in \netvic$.
\end{proposition}
\begin{proof}
  First, we reduce $\gamma\gamma$ redexes as much as possible, obtaining
  $\netone \red^* \nettwo$.
  By Lem.~\ref{lem:infContPreservation}, $\tkm{\nettwo}$ also has an infinite standard computation
  which is a continuation, say $\fcone$.
  Since $\fcone$ is a continuation,
  by definition it does not contain any marriage transition.
  Thus the number of marriage tokens is finitely bounded, implying
  there exists at least one token that moves infinitely many times.
  Let us follow the token's infinite path.
  The token starts from a principal port of a $\delta_2$ cell with an empty configuration.
  then it can\emph{not} hit
  \begin{varitemize}
  \item a principal port (of any kind of cells) since then a marriage happens.
  \item an auxiliary port of a $\delta_2$ cell, since $\fcone$ does not contain a marriage and thus it stops.
  \end{varitemize}
  So it must hit an auxiliary port of a $\gamma$ cell and goes out from its principal port, now with a non-empty configuration.
  After passing a $\gamma$ cell, it can\emph{not} hit
  \begin{varitemize}
  \item a principal port of a $\delta_2$ cell or an $\varepsilon$ cell since it stops due to $\gamma$ stack being not empty.
  \item a principal port of $\gamma$ cell since then $\nettwo$ has a $\gamma\gamma$ redex.
  \item an auxiliary port of a $\delta_2$ cell, because $\fcone$ does not contain a marriage and thus it stops.
  \end{varitemize}
  So again it must hit an auxiliary port of a $\gamma$ cell and goes out from its principal port with a non-empty configuration.
  Since $\nettwo$ is finite, the net $\nettwo$ necessarily has a loop
  consisting only of $\gamma$ cells connected with a tree containing
  a $\delta_2$ cell, i.e.\ a $\gamma$-vicious circle.
\end{proof}

\subsection{The Four Implications}
\begin{lemma}
$\infcomp{\netset}{\netvic}{\netone}\Longrightarrow\infcomp{\tkm{\netone}}{\emptyset}{\isone_{\tkm{\netone}}}$
\end{lemma}
\begin{proof}
  If $\infcomp{\netset}{\netvic}{\netone}$, then we can distinguish
  two cases:
  \begin{varitemize}
  \item
    Either $\netone$ admits an infinite reduction sequence. In this
    case, we prove that $\infcomp{\tkm{\netone}}{\emptyset}{\isone_{\tkm{\netone}}}$
    by coinduction. Indeed, if $\netone$ admits an infinite reduction
    sequence, then there is a finite reduction
    $$
    \netone\red_{\gamma\gamma}^{*}
    \nettwo\red_a\netthree
    $$
    where
    $a\in\{\delta\gamma,\delta\delta\}$ and
    $\infcomp{\netset}{\netvic}{\netthree}$. We know, by soundness,
    that $\tkm{\netone}\approx\tkm{\nettwo}$ and that $\isone_\nettwo$
    reduces in \emph{at least} one step to a state $\istwo$ which is
    bisimilar to $\isone_\netthree$. We can then conclude by
    coinduction that since
    $\infcomp{\tkm{\netthree}}{\emptyset}{\isone_{\tkm{\netthree}}}$,
    it holds that
    $\infcomp{\tkm{\netone}}{\emptyset}{\isone_{\tkm{\netone}}}$.
  \item
    Or $\netone$ admits a finite reduction sequence to a net
    $\nettwo$ in $\netvic$. The net $\nettwo$ contains a
    $\gamma$-vicious circle, and then
    $\infcomp{\tkm{\nettwo}}{\emptyset}{\isone_{\tkm{\nettwo}}}$ holds.
    By induction on the number of reduction steps leading $\netone$ to
    $\nettwo$, one can then prove that
    $\infcomp{\tkm{\netone}}{\emptyset}{\isone_{\tkm{\netone}}}$.
  \end{varitemize}
  This concludes the proof.
\end{proof}

\begin{lemma}
  $\fincomp{\netset}{\netvic}{\netone}\Longrightarrow\fincomp{\tkm{\netone}}{\emptyset}{\isone_{\tkm{\netone}}}$
\end{lemma}
\begin{proof}
  If  $\fincomp{\netset}{\netvic}{\netone}$, then $\netone$
  reduces in $n$ steps to a normal net not in $\netvic$. One can prove,
  that, for any such net $\netone$, $\tkm{\netone}$ admits
  a finite computation:
  \begin{varitemize}
  \item
    First of all, if $\netone$ is itself a normal form not in $\netvic$,
    then $\netone$ admits a finite (standard) computation, due to
    Proposition \ref{prop:netnormalform}.
  \item
    If $\netone\red\nettwo$ and $\tkm{\nettwo}$ admits a finite
    standard computation, then $\tkm{\netone}$ itself admits a finite
    standard computation, due to Proposition~\ref{prop:prescompone}
    and Proposition~\ref{prop:prescomptwo}.
  \end{varitemize}
\end{proof}

\begin{lemma}\label{lem:infMachineToInfNet}
$\infcomp{\tkm{\netone}}{\emptyset}{\isone_{\tkm{\netone}}}
\Longrightarrow\infcomp{\netset}{\netvic}{\netone}$
\end{lemma}
\begin{proof}
  If there is an infinite computation from $\isone_{\tkm{\netone}}$, then:
    \begin{varitemize}
    \item
      Either there is an infinite standard computation with a finite
      number $n$ of thunks and an infinite continuation from $\isone_{\tkm{\netone}}$.
      Then, by induction on $n$, one can prove that $\netone$ reduces
      to another net $\nettwo$ such that $\tkm{\nettwo}$ has an
      infinite standard reduction which is a continuation.
      By Proposition~\ref{prop:infCompImpliesVC}, this implies the thesis.
    \item
      Or there is an infinite standard computation with infinitely
      many thunks from $\isone_{\tkm{\netone}}$. Then, we can apply
      Proposition \ref{prop:finstdsucc} in a coinductive way to
      get an infinite reduction from $\netone$.
    \end{varitemize}
\end{proof}

\begin{lemma}
  $\fincomp{\tkm{\netone}}{\emptyset}{\isone_{\tkm{\netone}}}\Longrightarrow\fincomp{\netset}{\netvic}{\netone}$
\end{lemma}
\begin{proof}
  If there is a finite computation from $\isone_{\tkm{\netone}}$,
  then there is also a finite standard computation with $n$ thunks
  from $\isone_{\tkm{\netone}}$.
  We can then proceed by induction on $n$,
  thanks to Proposition~\ref{prop:finstdzerofinite}
  and Proposition~\ref{prop:finstdsucc}.
  If $n = 0$ then $\netone$ is in normal form not in $\netvic$ by Proposition~\ref{prop:finstdzerofinite}; if $n > 0$ then there exists $\nettwo$ satisfying $\netone \red^* \nettwo$ and $\tkm{\nettwo}$ admits a finite standard computation with $n-1$ thunks by Proposition~\ref{prop:finstdsucc}.
\end{proof}
\section{Discussion}\label{sect:discussion}
A formal comparison between multi-token machines and other concurrent
models of computation, and in particular the so-called truly
concurrent ones, is outside the scope of this paper. Some observations
in this direction are anyway useful, and are the starting point of
current investigations by the authors.

Multi-token machines as we defined them in this paper can be seen as
Petri nets where, however, the set of \emph{places} and
\emph{transitions} are both denumerable. In particular, places
correspond to the possible states of a token, of which there can be
countably many: remember that marriage tokens carry two stacks of
unbounded length with them. Indeed, tokens lying at the same wire, but with
different stacks need to correspond to different places at the level
of Petri nets, because they can travel in completely different
directions. Implementing the marriage mechanism in
Petri nets is indeed possible, but requires \emph{nonlocal}
transitions, transitions with incoming arcs coming from places which
are ``far away'' from each other: actually, there must be one such
transition for any pair of (virtual copies of) $\delta_2$ principal
ports in the underlying net. In other words, the nonlocal behaviour of
certain multi-token machine interactions would be reflected, at the
level of Petri nets, by a possibly very complex and nonlocal system of
transitions. Apart from that, the encoding seems possible, and indeed
targets plain Petri nets rather than the kind of Petri nets that are
usually required to model process algebras like the $\pi$-calculus,
e.g., Petri nets with inhibitor arcs~\cite{BusiGorrieri2009}. Whether
a finitary $\pi$-calculus can be given a plain Petri net semantics this
way is an interesting open problem.

Perhaps even more direct is the correspondence between multi-token machines
and chemical abstract
machines as defined by, e.g., Berry and Boudol~\cite{BerryBoudol1992}. In
this view, tokens would become \emph{molecules}, and would be allowed to
float in a \emph{soup}, freely interacting between them. Marriage
transitions of multi-token machines would be modelled by chemical rules
allowing two matching molecules (i.e. marriage tokens) to interact,
producing some other tokens. Moving transitions, on the other hand,
are most often possible in presence of a status token, and this
mechanism itself can be seen as a chemical reaction (which leave one
of the two involved ingredients essentially unchanged). In this view,
the underlying net's wires would become molecule kinds, and the
obtained abstract machine would be syntax-free.
\section{Conclusions}
In this paper, we introduced the first truly concurrent geometry
of interaction model, by defining multi-token machines for Mazza's
multiport interaction combinators, and by proving them to be sound
and adequate. Mazza himself suggested this as an open problem \cite{MazzaThesis}.
What is interesting about our results is also the way they are spelled
out and proved, namely by way of labelled transition systems and (bi)similarity.
This is something quite natural, but new to geometry of interaction
models, at least to the authors' knowledge.

As already mentioned, one topic for future work is a study of
multi-token machines for differential interaction nets, and in
particular to a ``multiport'' variation of them, something the authors
are actually working on as a way towards a better modelling of the
$\pi$-calculus. Another topic of future work concerns the study of
the concurrent nature of multi-token machines, and in particular of
their causal structure, especially in view of their interpretation as
Petri nets.

\section*{Acknowledgment}
The first author is partially supported by the ANR projects 12IS02001
PACE, and 14CE250005 ELICA. The third author is partially supported
by the JST ERATO Grant Number JPMJER1603. The three authors are
supported by the INRIA-JSPS joint project CRECOGI.

\bibliographystyle{plain}
\bibliography{biblio}

\newpage
\section*{Appendix}
In this section, we will give some more details as for why all the
results we proved in this paper can be generalised to MIC where
$\varepsilon$ cells are allowed to interact with other cells
(including themselves).
\paragraph*{Interaction Rules.}
First of all, we need some additional interaction rules, namely those
in Figure~\ref{fig:epsilonreduction}. Please notice that this new
sets of rules taken in isolation turns MICs into a confluent and
strongly normalizing rewrite system. Actually, the first of the three new
rules implies that we cannot work with $\delta_2$ cells only, but also with
$\delta_1$ cells.

\begin{figure*}
  \begin{center}
  \fbox{
    \begin{minipage}{.97\textwidth}
      \centering
      \includegraphics[scale=0.9]{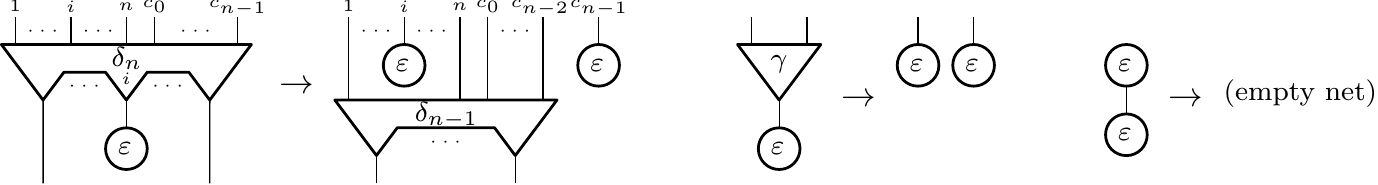}
      \caption{MIC Interaction Rules with $\varepsilon$}
      \label{fig:epsilonreduction}
  \end{minipage}}
  \end{center}
\end{figure*}

\paragraph*{Token Machines.}
As for token machines, we need to substantially generalise the ones we
introduced in the paper. First of all, cell types need to be
generalised. Moreover, more reduction rules are needed, because
$\delta$ cells can now marry not only with $\delta$ or $\gamma$ cells,
but also with $\varepsilon$ cells. They are in Figure~\ref{fig:epsiloninternal},
Figure~\ref{fig:epsilonmarriage}, and Figure~\ref{fig:epsilonexternal}.
The proof of the compositionality remains
essentially unaltered, except for the fact that more cases need to
be taken into account.

\begin{figure*}
  \begin{center}
  \fbox{
    \begin{minipage}{.97\textwidth}
      \centering
      \includegraphics[scale=0.9]{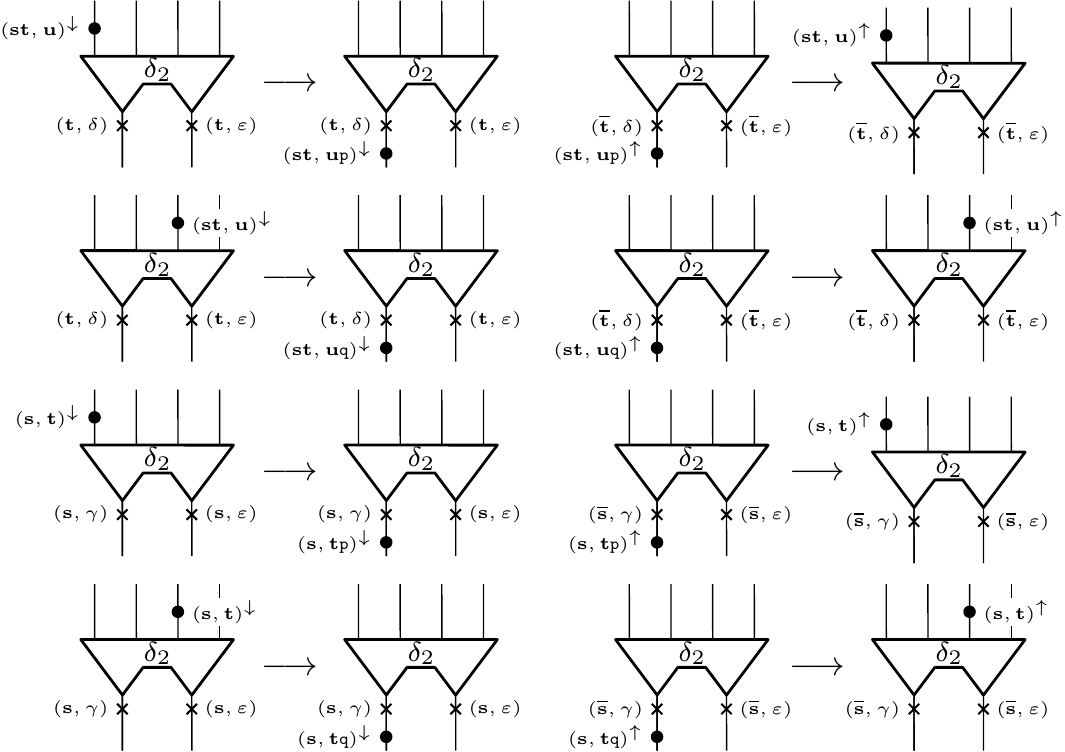}
      \caption{Internal Transition Rules with $\varepsilon$}
      \label{fig:epsiloninternal}
  \end{minipage}}
  \end{center}
\end{figure*}
\begin{figure*}
  \begin{center}
  \fbox{
    \begin{minipage}{.97\textwidth}
      \centering
      \includegraphics[scale=0.9]{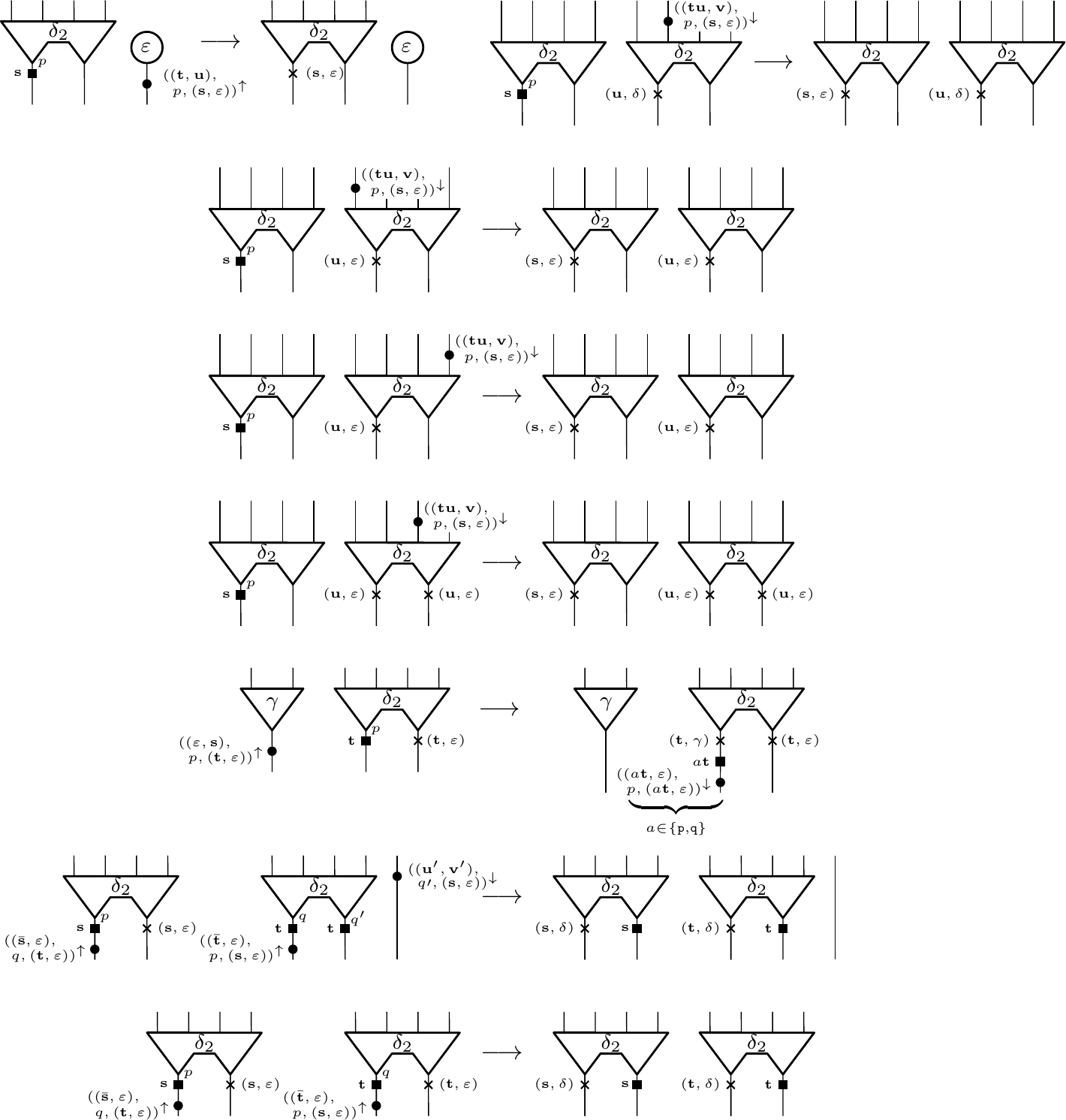}
      \caption{Internal Marriage Rules with $\varepsilon$}
      \label{fig:epsilonmarriage}
  \end{minipage}}
  \end{center}
\end{figure*}
\begin{figure*}
  \begin{center}
  \fbox{
    \begin{minipage}{.97\textwidth}
      \centering
      \includegraphics[scale=0.9]{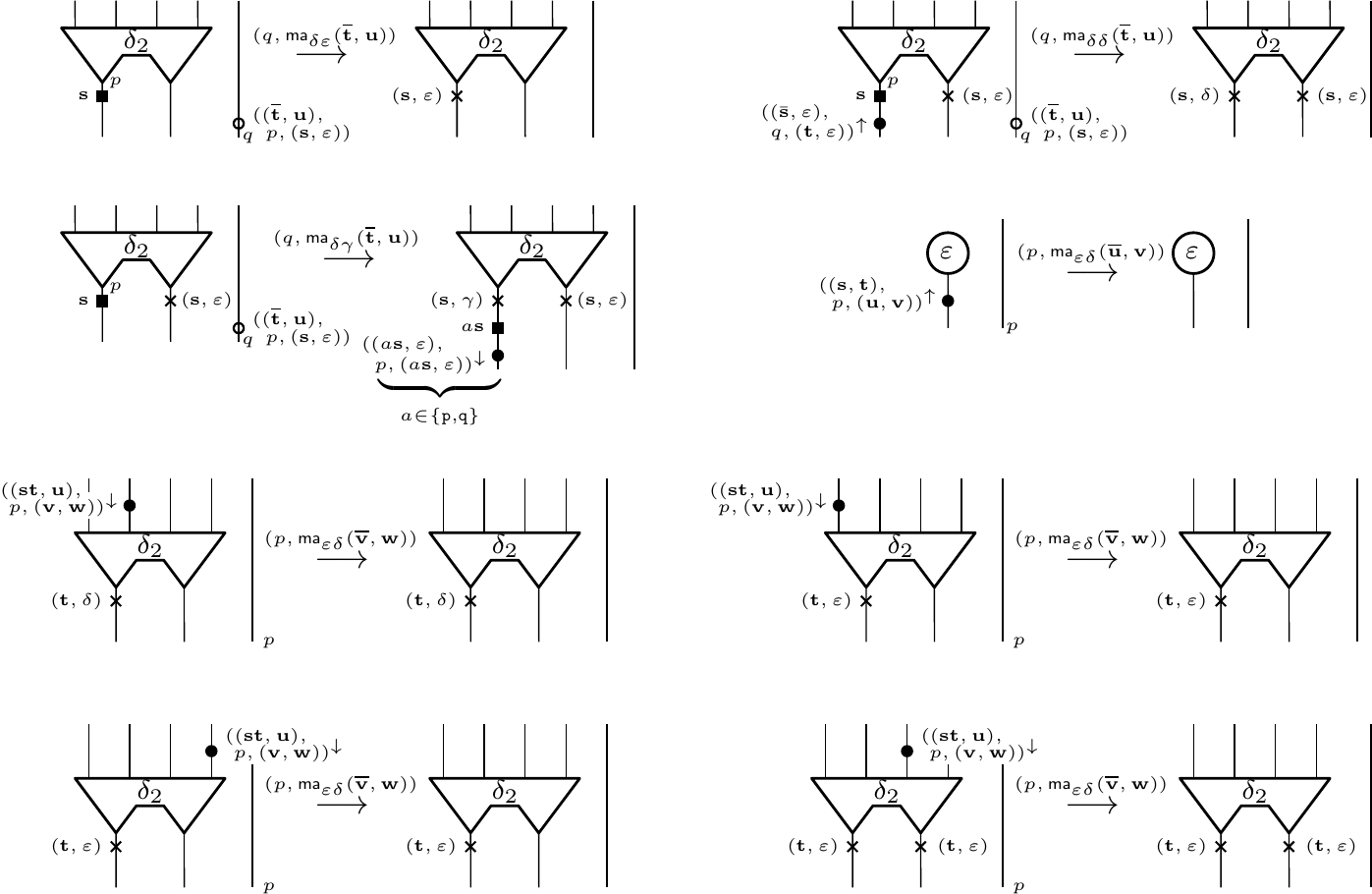}
      \caption{External Marriage Rules with $\varepsilon$}
      \label{fig:epsilonexternal}
  \end{minipage}}
  \end{center}
\end{figure*}

\paragraph*{Soundness.}
The proof of soundness needs to be adapted. In particular, we need
to prove lemmas analogous to Lemma \ref{lemma:gammagammasoundness},
Lemma~\ref{lemma:gammadeltasoundness} and Lemma~\ref{lemma:deltadeltasoundness},
but for interaction rules involving $\varepsilon$ cells. Here they
are:
\begin{lemma}\label{lemma:gammaepsilonsoundness}
  If $\netone\pred{\gamma\varepsilon}\nettwo$, then
  $\tkm{\netone}\approx\tkm{\nettwo}$.
\end{lemma}
\begin{lemma}\label{lemma:epsilonepsilonsoundness}
  If $\netone\pred{\varepsilon\varepsilon}\nettwo$, then
  $\tkm{\netone}\approx\tkm{\nettwo}$.
\end{lemma}
\begin{lemma}\label{lemma:deltaepsilonsoundness}
  If $\netone\pred{\delta\varepsilon}\nettwo$, then
  there is $\istwo_\netone$ such that $\isone_\netone\red\istwo_\netone$ and
  $(\fmsts{\tkns^\netone},\red,\istwo_\netone)\approx\tkm{\nettwo}$.
\end{lemma}
\begin{proof}
  We are in the following situation:
  \begin{center}
    \includegraphics[scale=1.0]{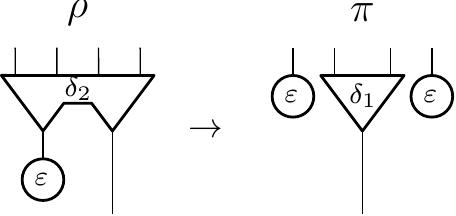}
  \end{center}
  Let $\istwo_\netone$ be a state which is obtained by performing a
  $\delta\varepsilon$ marriage in $\isone_\netone$. Then we define a bisimulation
  relation in the same way as Lemma~\ref{lemma:deltadeltasoundness}.
\end{proof}

\paragraph*{Adequacy.}
Finally, also the adequacy proof needs to be adapted, but only as for
some of the technical lemmas, the overall structure of the proof
remaining essentially unaltered. In particular, whenever we repeatedly
apply $\gamma\gamma$ in the proofs, we instead apply the rules
$\{\gamma\gamma, \gamma\varepsilon, \varepsilon\varepsilon\}$.  This
set of rules is also confluent and terminating, because the number of
cells (in particular $\gamma$ cells) is decreasing.

\end{document}